\documentclass[journal,twocolumn,twoside]{IEEEtran} 

\usepackage[T1]{fontenc}

\usepackage{graphicx}
\graphicspath{{./Figs/}}
\usepackage{subfigure}
\ifCLASSINFOpdf
\else
\fi

\usepackage{dblfloatfix}    

\usepackage[dvipsnames]{xcolor}

\def\P{\mathbb{P}}

\def\qt{\quad\times}
\def\qa{\quad+}
\def\qs{\quad-}

\def\s0{\mathrm{s}_0}
\def\Gtml{G_{\mathrm{t}}^{\mathrm{ml}}}
\def\Gtsl{G_{\mathrm{t}}^{\mathrm{sl}}}
\def\Gr{G_{\mathrm{r}}}
\def\Grv{G_{\mathrm{r,v}}}
\def\Grvmax{G_{\mathrm{r,v}}^{\mathrm{max}}}
\def\Grh{G_{\mathrm{r,h}}}

\def\themin{\theta_{\mathrm{min}}}

\def\Pvis{P_{\mathrm{vis}}}
\def\Pserml{P_{\mathrm{ser}}^{\mathrm{ml}}}
\def\Psersl{P_{\mathrm{ser}}^{\mathrm{sl}}}
\def\Pinv{P_{\mathrm{inv}}}
\def\Pout{P_{\mathrm{out}}}
\def\Poutml{P_{\mathrm{out}}^{\mathrm{ml}}}
\def\Poutsl{P_{\mathrm{out}}^{\mathrm{sl}}}

\def\Cml{\mathcal{C}_{\mathrm{ml}}}
\def\Csl{\mathcal{C}_{\mathrm{sl}}}
\def\mcV{\mathcal{V}}

\def\Psermlapp{\bar{P}_{\mathrm{ser}}^{\mathrm{ml}}}
\def\Pserslapp{\bar{P}_{\mathrm{ser}}^{\mathrm{sl}}}
\def\Pinvapp{\bar{P}_{\mathrm{inv}}}
\def\Poutapp{\bar{P}_{\mathrm{out}}}
\def\Poutmlapp{\bar{P}_{\mathrm{out}}^{\mathrm{ml}}}
\def\Poutslapp{\bar{P}_{\mathrm{out}}^{\mathrm{sl}}}

\def\PvoidBPP{P_{\mathrm{void}}^{\mathrm{BPP}}}
\def\PvoidPPP{P_{\mathrm{void}}^{\mathrm{PPP}}}

\def\re{r_{\mathrm{e}}}

\def\BPP{\Phi}
\def\BPPml{\Phi_{\mathrm{vis}}^{\mathrm{ml}}}
\def\BPPsl{\Phi_{\mathrm{vis}}^{\mathrm{sl}}}

\def\BPPvis{\Phi_{\mathrm{vis}}}
\def\BPPinv{\Phi_{\mathrm{inv}}}

\def\PPP{\bar{\Phi}}

\def\A{\mathcal{A}}
\def\Avis{\mathcal{A}_{\mathrm{vis}}}

\def\Aml{\mathcal{A}_{\mathrm{vis}}^{\mathrm{ml}}}
\def\Asl{\mathcal{A}_{\mathrm{vis}}^{\mathrm{sl}}}

\def\SBPP{\mathcal{S}({\A})}
\def\SBPPvis{\mathcal{S}({\Avis})}
\def\SBPPml{\mathcal{S}({\Aml})}
\def\SBPPsl{\mathcal{S}({\Asl})}

\def\psmax{\psi_{\mathrm{max}}}
\def\psth{\psi_{\mathrm{th}}}

\def\dmax{d_{\mathrm{max}}}
\def\dth{d_{\mathrm{th}}}

\def\wth{\omega_{\mathrm{th}}}

\usepackage{amssymb}
\usepackage{amsmath}
\usepackage[cmintegrals]{newtxmath}
\usepackage{stackengine}
\def\delequal{\mathrel{\ensurestackMath{\stackon[1pt]{=}{\scriptstyle\Delta}}}}

\usepackage{algorithmic}
\usepackage{algorithm}

\usepackage{amsthm}
\usepackage{xpatch}

\newtheorem{thm}{Theorem}
\newtheorem{lem}{Lemma}

\newtheorem{cor}{Corollary}
\newtheorem{rem}{Remark}

\xpatchcmd{\proof}{\hskip\labelsep}{\hskip5\labelsep}{}{}

\usepackage{mathtools}

\newcounter{myeqncount} 

\usepackage{array}

\begin{document}
\title{Performance Analysis of Satellite Communication System Under the Shadowed-Rician Fading: \\A Stochastic Geometry Approach}

\author{Dong-Hyun Jung, Joon-Gyu Ryu, Woo-Jin Byun, and Junil Choi\\
\thanks{This work was supported by Institute of Information \& communications Technology Planning \& Evaluation (IITP) grant funded by the Korea government (MSIT) (No.2021-0-00847, Development of 3D Spatial Satellite Communications Technology).}
\thanks{D.-H. Jung is with the School of Electrical Engineering, KAIST, and with the Radio and Satellite Research Division, Communication and Media Research Laboratory, Electronics and Telecommunications Research Institute, Daejeon, South Korea (e-mail: donghyunjung@kaist.ac.kr).}
\thanks{J.-G. Ryu and W.-J. Byun are with the Radio and Satellite Research Division, Communication and Media Research Laboratory, Electronics and Telecommunications Research Institute, Daejeon, South Korea (e-mail: \{jgryurt, wjbyun\}@etri.re.kr).}
\thanks{J. Choi is with the School of Electrical Engineering, KAIST, Daejeon, South Korea (e-mail: junil@kaist.ac.kr).}
\vspace{-0.5cm}
}
\maketitle

\begin{abstract}
In this paper, we consider downlink low Earth orbit (LEO) satellite communication systems where multiple LEO satellites are uniformly distributed over a sphere at a certain altitude according to a homogeneous binomial point process (BPP). Based on the characteristics of the BPP, we analyze the distance distributions and the distribution cases for the serving satellite. We analytically derive the exact outage probability, and the approximated expression is obtained using the Poisson limit theorem. With these derived expressions, the system throughput maximization problem is formulated under the satellite-visibility and outage constraints. To solve this problem, we reformulate it with bounded feasible sets and propose an iterative algorithm to obtain near-optimal solutions. Simulation results perfectly match the derived exact expressions for the outage probability and system throughput. The analytical results of the approximated expressions are fairly close to those of the exact ones. It is also shown that the proposed algorithm for the throughput maximization is very close to the optimal performance obtained by a two-dimensional exhaustive search.

\\
\end{abstract}

\IEEEpeerreviewmaketitle

\section{Introduction}\label{sec:intro}
\IEEEPARstart{S}{atellite} communications have recently attracted significant attention as a solution to provide global coverage without deploying base stations, which requires high cost.
The 3rd Generation Partnership Project (3GPP) is trying to include non-terrestrial networks (NTNs) as a part of the fifth generation (5G) standard, which considers flying objects as entities in 5G networks such as geostationary orbit (GEO) satellites, low Earth orbit (LEO) satellites, and high altitude platform stations (HAPSs) [\ref{Ref:3GPP_38.811}]. 
The goal of the standardization is to integrate satellites and HAPSs into terrestrial networks (TNs) in order to provide communication services to both terrestrial users without any infrastructure nearby and flying objects such as airplanes, drones, and vehicles for urban air mobility.

There are some challenges to directly integrate satellites into the TNs.
The high altitudes of satellites, e.g., $35,786$ km for GEO satellites and $300 - 2,000$ km for LEO satellites, cause long propagation delays.
In addition, the large beam coverage, which is often considered as an advantage of the NTNs, makes different round trip delays between the nearest and the farthest terminals from a satellite. 
In the initial access procedure, the large amount of timing difference may require a larger preamble receiving window and a longer period of random access channel occasions for timing synchronization.
A large amount of Doppler shift and drift is another big challenge for NTNs since the satellites have to move along the orbits with a certain velocity, e.g., more than $25,000$ km/s for LEO satellites, to keep the orbit against the Earth's gravity.
Such problems have been actively investigated in the 3GPP standard [\ref{Ref:3GPP_38.821}].

\subsection{Related Works}\label{sec:intro:relatedWork}
Integration between satellite and terrestrial networks has been investigated in [\ref{Ref:Lin1}]-[\ref{Ref:Zhen}].
A beamforming scheme for cognitive satellite-terrestrial networks was proposed in [\ref{Ref:Lin1}] where a base station and a cooperative terminal are exploited to enhance the secrecy performance.
The sum rate maximization problem for satellite and aerial-integrated terrestrial networks was solved in [\ref{Ref:Lin2}] where multicast communications with rate-splitting multiple access were considered. 
A secure beamforming scheme for cognitive satellite-terrestrial networks was proposed in [\ref{Ref:Lin3}] to maximize secrecy energy efficiency.
Satellite-integrated 5G networks were specifically considered in [\ref{Ref:Guidotti}]-[\ref{Ref:Zhen}].
Technical challenges and impacts of the satellite channel characteristics on the physical and medium access control layers were discussed based on the 3GPP NTN architecture in [\ref{Ref:Guidotti}].
A load balancing algorithm was proposed for multi-radio access technology networks including non-terrestrial and terrestrial networks in [\ref{Ref:Shahid}].
A new preamble design for random access and a preamble detection scheme were proposed in [\ref{Ref:Zhen}] to tolerate the large difference between the round trip delays the nearest and the farthest terminals experience from a serving satellite.

For modeling satellite channels, the shadowed-Rician fading model was proposed in [\ref{Ref:Loo}], which has been proved its suitability in various frequency bands, e.g., the UHF-band, L-band, S-band, and Ka-band.
A simpler model for the shadowed-Rician fading by using the Nakagami distribution for the amplitude of line-of-sight (LOS) component was proposed in [\ref{Ref:Abdi}] and was widely adopted to analyze the system performance of satellite communication systems [\ref{Ref:Jung}]-[\ref{Ref:An}].
The outage probability of shared-band on-board processing satellite communication systems was analyzed in [\ref{Ref:Jung}] under the shadowed-Rician fading where the satellite has digital processing capability.
In [\ref{Ref:Bhatnagar}], the approximated closed-form expressions for the probability density function (PDF) and cumulative distribution function (CDF) of received signal-to-noise power ratio (SNR) were analyzed under the shadowed-Rician fading where maximum ratio combining is used for multi-antenna reception.
The approximated bit error rate and outage probability of decode-and-forward relaying-based satellite communication systems were derived in [\ref{Ref:Bhatnagar2}] where a source and a destination are equipped with multiple antennas, each experiences independent and identically distributed shadowed-Rician fading.
For satellite-terrestrial relay networks, the outage probability and ergodic capacity were analyzed in [\ref{Ref:Bankey}] with opportunistic user scheduling, and the ergodic capacity was studied in [\ref{Ref:An}] for two adaptive transmission schemes.
Although the shadowed-Rician channel model is appropriate for satellite channels in various frequency bands, there are not many system-level analyses under the shadowed-Rician fading.

Stochastic geometry is a popular analytical tool for estimating system-level performance of communication systems [\ref{Ref:Andrews}]-[\ref{Ref:Kolawole}].
In stochastic geometry-based analyses, nodes are usually distributed according to Poisson point processes (PPPs), i.e., the number of nodes is randomly determined by the Poisson distribution, and the nodes are uniformly located on the infinite two-dimensional area. 
However, when the number of nodes distributed in the networks is finite, e.g., satellite networks, the randomness of the positions of nodes should be modeled by using a finite point process other than the PPP [\ref{Ref:Guo}].

Binomial point process (BPP) is a finite point process that describes the distribution of the finite number of points on a finite area where each point exists in a certain area by the binomial distribution [\ref{Ref:Book:Chiu}].
The positions of LEO satellite constellations can be modeled as the BPP at a certain altitude because the satellites may look randomly distributed due to their fast mobility and various types of orbits [\ref{Ref:Okati}]-[\ref{Ref:Talgat2}]. The work [\ref{Ref:Okati}] showed that the BPP well models practical LEO satellite constellations from coverage and rate perspectives.
In [\ref{Ref:Talgat}], the user coverage probability of LEO satellite communication systems was studied where gateways act as relays between users and LEO satellites.
The distance distributions for gateway-satellite and inter-satellite links were studied in [\ref{Ref:Talgat2}].
However, in these works, the exact outage probability and throughput of LEO satellite communication systems using the stochastic geometry have not been analyzed under the shadowed-Rician fading.

\begin{figure*}[!t]
\centering
\subfigure[]{
\includegraphics[width=0.95\columnwidth]{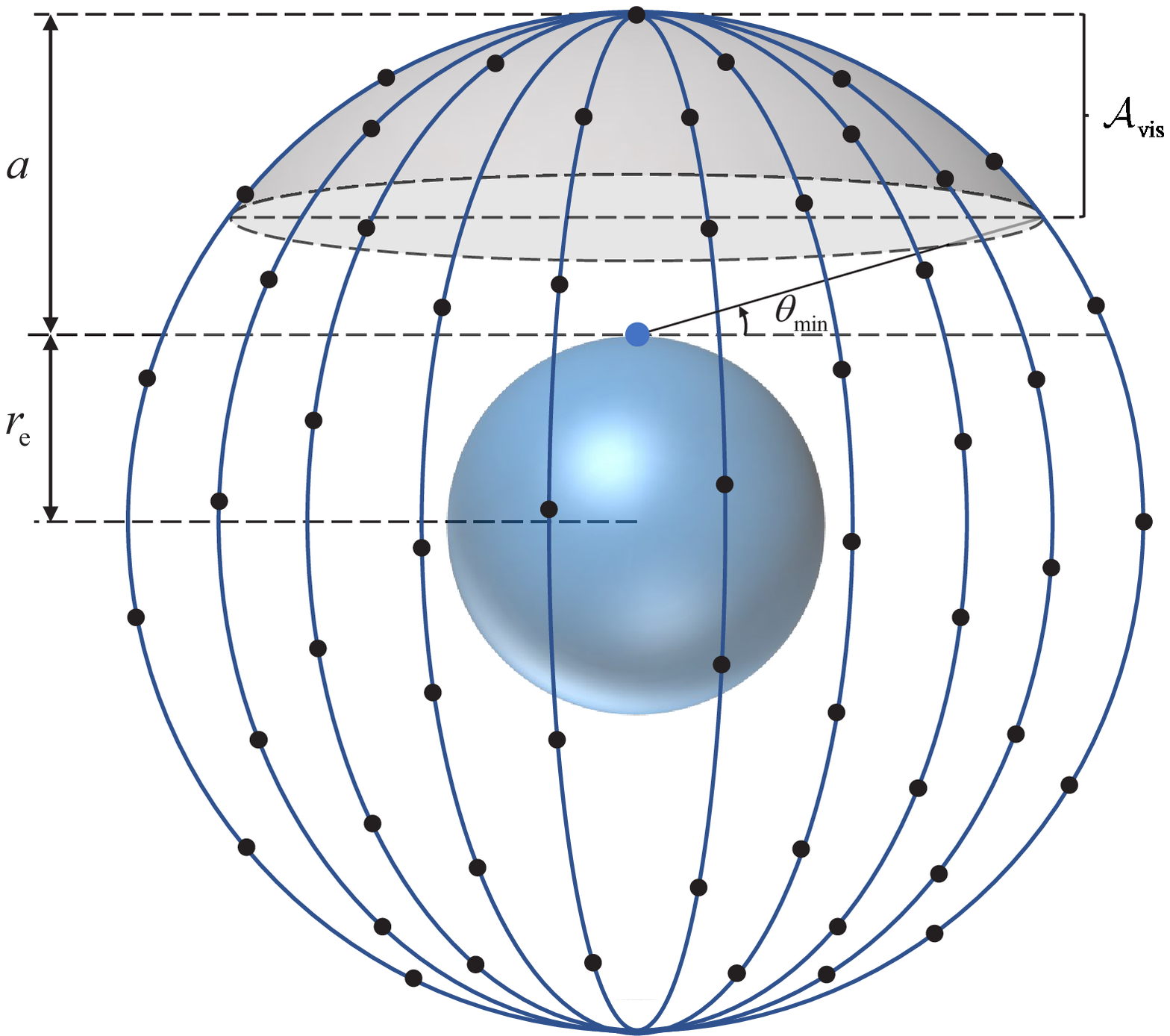}
\label{Fig:System_model_1}
}
\subfigure[]{
\includegraphics[width=0.95\columnwidth]{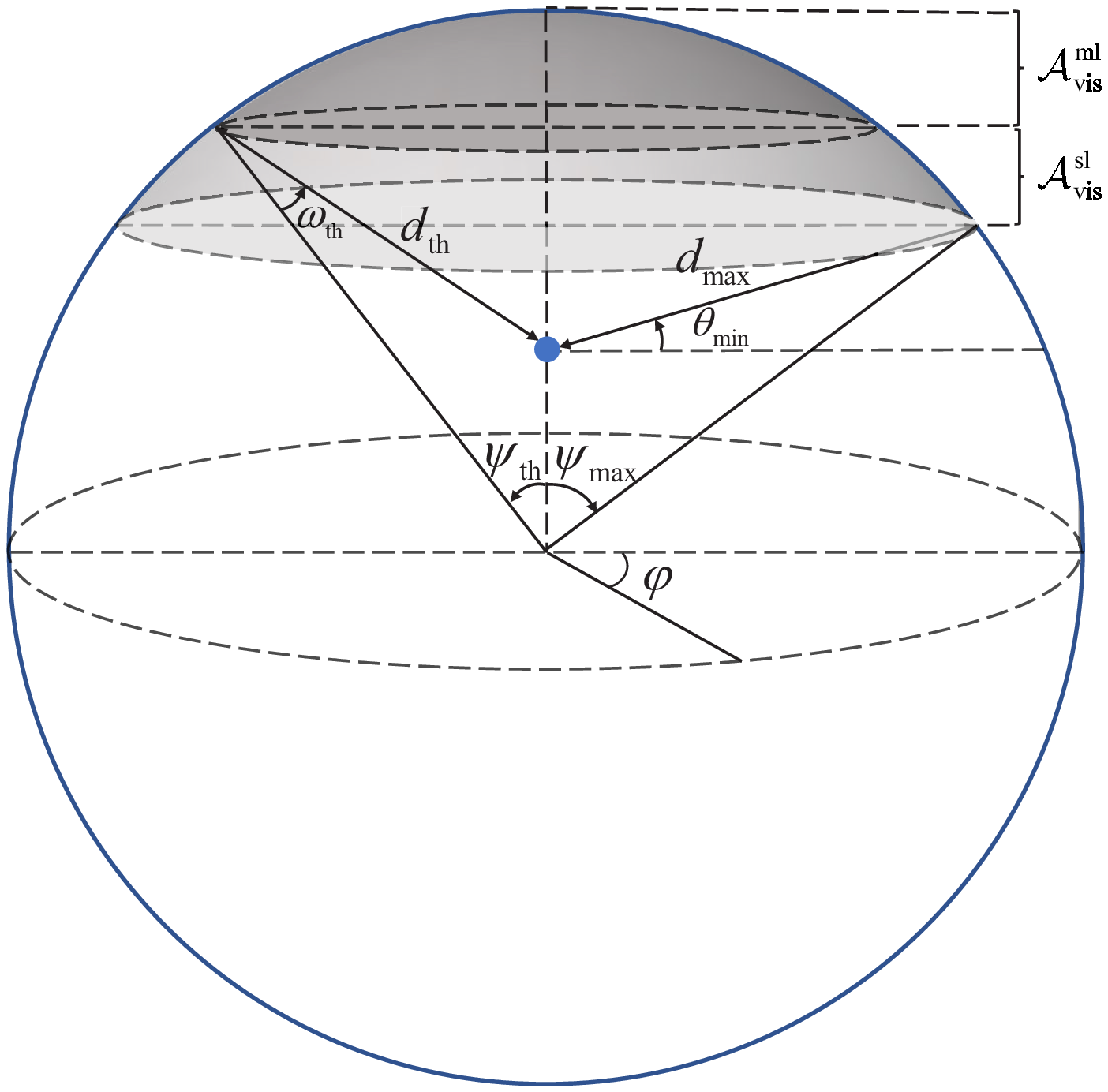}
\label{Fig:System_model_2}
}
\caption{(a) System model and (b) parameter description. Blue sphere, blue dot, and black dots indicate the Earth, the terminal, and the satellites, respectively.}
\label{Fig:System_model}
\end{figure*}

\subsection{Contributions}\label{sec:intro:Cont}
In this paper, we consider a downlink LEO satellite communication system where multiple satellites are distributed according to a homogeneous BPP. Under the shadowed-Rician fading, the exact performance of system is derived and the approximated performance is provided for mathematical tractability.
The main contributions of the paper are summarized as follows:
\begin{itemize}
    \item We adopt directional beamforming with fixed beam antennas for the satellites that maintain the boresight fixed in the direction of the subsatellite point (the nearest point on the Earth). For mathematical tractability, we use sectorized beam patterns with two sectors: main and side lobes with constant gains. Different from [\ref{Ref:Okati}] and [\ref{Ref:Talgat}], we consider the impact of the beam pattern on the system performance.
    \item We derive the distributions of three distances: (i) distance to the nearest satellite, (ii) distance to the serving satellite whose main lobe is directed to the terminal, and (iii) distance to the serving satellite whose side lobe is directed to the terminal. These distance distributions are a key to perform the stochastic geometry-based analyses.
    \item We analyze the three distribution cases for the serving satellites. The probabilities of these cases are derived using the void probability of the BPP. 
    These probabilities are essential to derive the system performance but were not considered in [\ref{Ref:Talgat}].
    \item We derive the exact expressions for the outage probability and system throughput in closed forms, considering the distribution cases for the serving satellite. We also derive the approximated ones using the Poisson limit theorem.
    \item With the derived expressions, we formulate the system throughput maximization problem under the satellite-visibility and outage constraints. To solve the problem, we reformulate it with bounded feasible sets and propose an iterative algorithm. We also analyze the computational complexity of the proposed algorithm with both exact and approximated expressions. 
    \item  Finally, we numerically show that the derived exact expressions perfectly match Monte-Carlo simulations, and the approximated ones are also fairly close. It is also shown that the proposed algorithm for the throughput maximization has close performance to the optimal solutions.
\end{itemize}

The rest of this paper is organized as follows.
In Section \ref{sec:model}, the system and channel models for a satellite communication system are described.
In Section \ref{sec:distDist}, the surface areas and distance distributions are analyzed using the characteristics of the BPP.
In Section \ref{sec:exact}, we derive the exact outage probability and system throughput in closed forms.
In Section \ref{sec:app}, we obtain approximated expressions for the system performance based on the Poisson limit theorem.
In Section \ref{sec:thruMax}, we propose an iterative algorithm for throughput maximization and compute the computational complexity.
In Section \ref{sec:sim}, simulation results are provided, and conclusions are drawn in Section~\ref{sec:con}.

\textbf{Notation:}
$\P[\cdot]$ indicates the probability measure.
The CDF and the PDF of a random variable $X$ are $F_X(x)$ and $f_x(x)$, respectively.
The empty set is $\emptyset$, and the complement of a set $\mathcal{X}$ is $\mathcal{X}^{\mathrm{c}}$.
The surface area of a region $\mathcal{X}$ is $\mathcal{S}({\mathcal{X}})$.
$\Gamma(\cdot)$ is the Gamma function, and the Pochhammer symbol is defined as $(x)_n=\Gamma(x+n)/\Gamma(x)$.
The lower incomplete Gamma function is defined as $\gamma(a, x)=\int_0^x t^{a-1}\exp(-t)dt$.
$\binom{n}{k}$ denotes the binomial coefficient.
The inverse function of $f(\cdot)$ is $f^{-1}(\cdot)$.

\section{System and Channel Models}\label{sec:model}
Consider a downlink LEO satellite communication system with $S$ satellites at altitude $a$ communicating with terminals on the Earth as shown in Fig. \ref{Fig:System_model_1}.
We assume that $S$ satellites are uniformly distributed on a surface of a sphere with the radius $\re+a$ according to a homogeneous BPP ${\BPP}$ with density $\lambda_{\mathrm{s}}$ where $\re$ is the radius of the Earth [\ref{Ref:Okati}].
The surface where the satellites are distributed can be expressed with spherical coordinates as $\A=\{\rho=\re+a, 0\le\psi\le\pi, 0 \le\varphi\le 2\pi\}$ where $\rho$, $\psi$, and $\varphi$ are the radial distance, polar angle, and azimuthal angle, respectively.
Assume that the satellites work as base stations where the satellites are connected to core networks via wireless backhaul [\ref{Ref:3GPP_38.821}], and a terminal is associated with one of \textit{visible satellites}.
The visible satellites are located above the minimum elevation angle $\themin$, i.e., a pre-defined elevation angle above which the terminal can be served by a satellite.
The distance between the terminal and any visible satellite should be less than the maximum distance $\dmax$, which is obtained as $\dmax = \sqrt{\re^2\sin^2\themin + a^2 + 2 \re a}-\re\sin\themin$ by the law of cosines, 
\begin{align}\label{eq:law_of_cosines}
(\re+a)^2 = \dmax^2 + \re^2 + 2 \dmax \re \sin\themin.
\end{align}
The BPP of the satellites $\BPP$ can be divided into two sets: a set of the visible satellites $\BPPvis$ and a set of the invisible satellites $\BPPinv$.
The surface area where the satellites in $\BPPvis$ can be located is a spherical cap, shown as the shaded area in Fig. \ref{Fig:System_model_1}, which can be expressed as $\Avis = \{\rho=\re+a, 0\le\psi\le\psmax, 0 \le\varphi\le 2\pi\}$ where the maximum polar angle $\psmax$, below which the terminal can see the satellites, is obtained by the law of cosines as 
\begin{align}
\psmax
    =\cos^{-1}\left(\frac{\re^2+(\re + a)^2-\dmax^2}{2 \re (\re + a)}\right).    
\end{align}

As considered in the 3GPP NTN standard [\ref{Ref:3GPP_38.821}], two types of the terminals are assumed: (i) very-small-aperture terminal (VSAT) and (ii) handheld terminal, operated in Ka and S-bands, respectively.\footnote{In the Ka-band, to compensate the large path-loss and the rain attenuation, the VSAT terminals with several-meters antenna are typically used, while in the S-band, the light and portable handheld terminals are preferable.}
Without loss of generality, we analyze the downlink performance of a typical terminal located at a fixed position [\ref{Ref:Andrews}].
It is assumed that the interference from other satellites is negligible at the terminal thanks to interference management techniques such as frequency reuse and beamforming techniques [\ref{Ref:Talgat}].

Directional beamforming with fixed-beam antennas is adopted at the satellites, i.e., the satellites maintain the boresight of their beams in the direction of the subsatellite point.
Tapered-aperture antennas are used to model practical beam patterns of satellites as in [\ref{Ref:Huang1}]-[\ref{Ref:Zhang}].
However, for mathematical tractability, we assume that the satellites have sectorized beam patterns\footnote{The sectorized beam patterns are simplified versions of the practical beam patterns where they were widely adopted for theoretical analyses using the stochastic geometry [\ref{Ref:Alkhateeb}]-[\ref{Ref:Dabiri}].} 
where the antenna gains of the main and side lobes are $\Gtml$ and $\Gtsl$, respectively.
Let $\omega_s$, $s \in \BPPvis$, denote the angle between the terminal and the boresight direction of the satellite $s$.
Then, the transmit antenna gain of the satellite~$s$ is given by
\begin{align}
G_{\mathrm{t},s}
    =& \begin{cases} 
    \Gtml, & \mbox{if  } |\omega_s| \le \wth,\\
    \Gtsl, & \mbox{otherwise} 
    \end{cases}
\end{align}
where $\wth$ is the threshold angle between the main and side lobes of the beam pattern.
We assume that the VSAT terminal has a directional antenna with the gain $\Grv$, while the handheld terminal has an omnidirectional antenna with the gain $\Grh$.
The VSAT terminal attempts to track the serving satellite's trace for antenna beam-pointing but there may be a pointing error $\omega_{\mathrm{e}}$, i.e., the difference between the boresight and the direction to the serving satellite.
Then, the receive antenna gain of the VSAT is given by [\ref{Ref:Zhang}]
\begin{align}
\Grv 
    =& \begin{cases} 
        \Grvmax, & \mbox{if   } 0^\circ \le \omega_{\mathrm{e}} < 1^\circ,\\ 
        10^{3.2-2.5 \log \omega_{\mathrm{e}}}, & \mbox{if   } 1^\circ \le \omega_{\mathrm{e}} < 48^\circ,\\
        0.1, & \mbox{if   } 48^\circ \le \omega_{\mathrm{e}} < 180^\circ
     \end{cases}
\end{align}
where $\Grvmax$ is the maximum receive antenna gain.

For the Ka-band, the rain attenuation is usually modeled as lognormal distribution.
However, since we only focus on a typical terminal at a fixed position, the rain fading that all satellites experience in the satellite-terminal links is assumed to be identical and constant [\ref{Ref:Talgat}], [\ref{Ref:Na}], [\ref{Ref:Zheng}].
Thus, the rain attenuation for the Ka-band is given by $g_s=g$ for all satelllites $s \in \BPPvis$, while for the S-band, the rain attenuation is negligible, i.e., $g_s=g=1$.

The shadowed-Rician fading is assumed for the channels between the terminal and satellites, which is widely adopted for satellite channels in both S and Ka-bands [\ref{Ref:Jung}]-[\ref{Ref:Bhatnagar2}]. 
Let $h_{s}$ denote the channel gain between the terminal and the satellite~$s$. Then, the CDF of the channel gain is given by~[\ref{Ref:Abdi}]
\begin{equation}\label{eq:CDF_ch_gain}
{F_{h_{s}}}(x) = K\sum\limits_{n = 0}^\infty  {\frac{{{{(m)}_n}{\delta ^n}{{(2b)}^{1 + n}}}}{{{{(n!)}^2}}}}\gamma\left(1+n,\frac{x}{2b}\right)
\end{equation}
where $K = {\left({2bm}/{(2bm + \Omega) }\right)^m}/{2b}$, $\delta  = (\Omega /(2bm + \Omega ))/2b$ with $\Omega$ being the average power of LOS component, $2b$ is the average power of the multi-path component except the LOS component, and $m$ is the Nakagami parameter. 
We also assume that Doppler shifts caused by fast mobility of LEO satellites can be perfectly compensated using proper estimation techniques based on the satellite ephemeris information, e.g., the types of orbits, altitudes, positions, and velocity of satellites, which can be accurately known in prior [\ref{Ref:Arti}],~[\ref{Ref:Guo2}].

Let $d_s$ be the distance between the terminal and the satellite~$s$. 
Then, the path-loss between the terminal and the satellite~$s$ is given by
$\ell(d_s) =  \left(\frac{c}{4\pi f_{\mathrm{c}}}\right)^2 d_s^{-\alpha}$
where $c$ is the speed of light, $f_{\mathrm{c}}$ is the carrier frequency, and $\alpha$ is the path-loss exponent.
For the signals transmitted from the satellite $s$, the SNR at the terminal is given by
$\gamma_s =  \frac{P g G_{\mathrm{t},s} \Gr h_s \ell(d_s)}{N_0 W}$ where $P$ is the transmit power of the satellite, $\Gr$ is the receive antenna gain of the terminal, i.e., $\Grv$ or $\Grh$, $N_0$ is the noise power spectral density, and $W$ is the bandwidth.

\section{Surface Areas and Distance Distributions}\label{sec:distDist}
In this section, we first divide the area where the visible satellites are located, $\Avis$, into two areas and then calculate the surface areas of interest.
We also obtain the distribution of the distance to the nearest satellite and that to the serving satellite.

\subsection{Surface Areas of Interest}\label{sec:distDist:surf}
The set of visible satellites $\BPPvis$ is further separated into two sets $\BPPml$ and $\BPPsl$ consisting of the visible satellites whose main and side lobes are directed towards the terminal, respectively.
The surface areas where the satellites in $\BPPml$ and $\BPPsl$ can be located, are denoted by $\Aml$ and $\Asl$, respectively, and shown in Fig. \ref{Fig:System_model_2}.
The threshold polar angle differentiating $\Aml$ and $\Asl$ is given by
\begin{align}
\psth
    = \sin^{-1}\left(\frac{\re+a}{\re} \sin\wth\right)-\wth.
\end{align}
The regions $\Aml$ and $\Asl$ are expressed with respect to $\psth$ as $\Aml = \{\rho=\re+a, 0\le\psi\le\psth, 0 \le\varphi\le 2\pi\}$ and $\Asl = \{\rho=\re+a, \psth\le\psi\le\psmax, 0 \le\varphi\le 2\pi\}$, respectively.

Now, we obtain the surface areas of $\A$, $\Avis$, $\Aml$, and $\Asl$ in order.
The region where all the satellites are located, $\A$, is a sphere with the radius $\re + a$ whose surface area is given by 
$
    \SBPP = 4\pi(\re+a)^2
$.
The region $\Avis$ is a spherical cap with the radius $\re + a$ whose height is calculated as
$a-\dmax \sin\themin\mathop=(\dmax^2-a^2)/(2\re)$ by using \eqref{eq:law_of_cosines}.
Since the surface area of a spherical cap $\mathcal{X}$ with radius $r$ and height $q$ is given by $\mathcal{S}(\mathcal{X})=2\pi r q$ [\ref{Ref:Book:Polyanin}], the surface area of $\Avis$ is given by
\begin{align}
    \SBPPvis =\frac{\pi(\re+a)(\dmax^2-a^2)}{\re}.
\end{align}
Similarly, the cap height of $\Aml$ is $(\re+a)(1 - \cos\psth)$, so the surface area of $\Aml$ is given by 
\begin{align}
    \SBPPml = 2\pi(\re+a)^2(1 - \cos\psth).
\end{align}
The surface area of $\Asl$ is the difference between $\SBPPvis$ and $\SBPPml$, i.e.,
$
    \SBPPsl = \SBPPvis - \SBPPml
$.

The impact of the main lobe's beamwidth on the $\SBPPml$ and $\SBPPsl$ can be seen from the derivative of $\psth$ with respect to $\wth$, which is given by 
\begin{align}\label{eq:der_psth}
\frac{d\psth}{d\wth} = \frac{(\re+a)\cos\wth}{\sqrt{\re^2-(\re+a)^2\sin^2\wth}}-1 > 0.
\end{align}
In the first term on the right-hand side of \eqref{eq:der_psth}, the numerator is always larger then the denominator for $a>0$, which means that $\psth$ is an increasing function of $\wth$. Thus, as the beamwidth of the main lobe increases, the $\SBPPml$ enlarges and $\SBPPsl$ shrinks.
These surface areas are used to obtain the probabilities of the serving satellite's distributions in Section~\ref{sec:exact:DistCases}.

\begin{figure}
\begin{center}
\includegraphics[width=0.9\columnwidth]{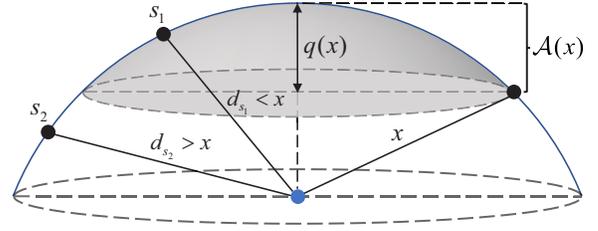}
\end{center}
\setlength\abovecaptionskip{.25ex plus .125ex minus .125ex}
\setlength\belowcaptionskip{.25ex plus .125ex minus .125ex}
\caption{Description of the area  $\mathcal{A}(x)$. $s_1$ and $s_2$ are the satellites whose distance from the terminal (blue dot) is greater than $x$ or less than $x$, respectively.}
\vspace{-0.2cm}
\label{Fig:area_description}
\end{figure}

\subsection{Distribution of Distance to Nearest Satellite}
We denote by $\mathcal{A}(x)$ a spherical cap with the radius $\re+a$ such that the distance between any point on $\mathcal{A}(x)$ and the terminal is less than $x\in [a, 2\re+a]$, which is shown as the shaded areas in Fig. \ref{Fig:area_description}.
By the Pythagorean theorem, the cap height of $\mathcal{A}(x)$ is calculated as $q(x)=(x^2-a^2)/(2\re)$.
The surface area of $\mathcal{A}(x)$ is given by
\begin{align}\label{eq:A_c}
\mathcal{S}(\mathcal{A}(x))
    = 2\pi(\re+a)q(x)
    = \frac{\pi(\re+a)(x^2-a^2)}{\re}.
\end{align}
For example, if $x=a$, the surface area vanishes, i.e., $\mathcal{S}(\mathcal{A}(a))=0$, while if $x=2\re+a$, the surface area becomes a whole sphere with the radius $\re+a$, i.e., $\mathcal{S}(\mathcal{A}(2\re+a))=4\pi(\re+a)^2$.
It is seen from \eqref{eq:A_c} that the surface area is an increasing function of $x$, meaning that the spherical cap becomes enlarged as $x$ increases.
Thus, the probability that the distance between the terminal and the satellite $s$ is less than $x$ is equivalent to the probability that the satellite $s$ is located in $\mathcal{A}(x)$, i.e., the success probability\footnote{Note that the success probability is the probability that a point is located on the area of interest. For homogeneous BPPs, the success probability is obtained as the ratio of the surface area of interest to the total surface area where all points are distributed [\ref{Ref:Book:Chiu}].} for $\mathcal{A}(x)$, which is obtained by the ratio of $\mathcal{S}(\mathcal{A}(x))$ to $\mathcal{S}(\mathcal{A})$ as
\begin{align}\label{eq:cdf_ds}
\P[d_s \le x] 
    = \frac{\mathcal{S}(\mathcal{A}(x))}{\mathcal{S}(\mathcal{A})}
    =\frac{x^2-a^2}{4\re (\re+a)}\delequal\kappa(x).
\end{align}

We denote by $D$ the distance between the terminal and the nearest satellite, of which CDF is obtained by using \eqref{eq:cdf_ds} as
\begin{align}\label{eq:CDF_D}
F_D(x)
    &=  1 - \P[D>x]\nonumber\\
    &=  1 - \prod_{s \in \BPP} \P[d_s>x]\nonumber\\
     &= \begin{cases} 
     0, & \mbox{if } x \le a,\\
    1 - \left(1-\kappa(x)\right)^S, & \mbox{if } a < x \le 2\re+a, \\
     1, & \mbox{if } x > 2\re+a.
     \end{cases}
\end{align}
By differentiating \eqref{eq:CDF_D}, the corresponding PDF is given by 
\begin{align}\label{eq:PDF_D}
f_D(x)
     = \begin{cases} 
     \frac{S x}{2\re(\re+a)} \left(1-\kappa(x)\right)^{S-1}, & \mbox{if } a < x \le 2\re+a, \\
     0, & \mbox{otherwise}.
     \end{cases}
\end{align}

\subsection{Distribution of Distance to Serving Satellite}
Let $Y$ denote the distance between the terminal and the serving satellite, given that the serving satellite is in $\Aml$, i.e., $\BPPml \neq \emptyset$.
Then, the CDF of $Y$ is given by 
\begin{align}\label{eq:CDF_Y1}
F_Y(x) 
    = \P[Y \le x]
    &= \P[D \le x|\BPPml \neq \emptyset]\nonumber\\
    &= \P[D \le x|D < \dth]
    = \frac{\P[D \le x,D < \dth]}{\P[D < \dth]}
\end{align}
where $\dth = \sqrt{\re^2 + (\re + a)^2 - 2 \re (\re + a) \cos\psth}$ is the distance between the terminal and the satellite located at the boundary between $\Aml$ and $\Asl$.
Using \eqref{eq:CDF_D}, the CDF becomes
\begin{align}\label{eq:CDF_Y}
F_Y(x) 
    = \begin{cases} 
    0, & \mbox{if  } x \le a,\\
    \frac{F_D(x)}{F_D(\dth)}, & \mbox{if  } a < x \le \dth,  \\
    1, & \mbox{if  } x > \dth.
    \end{cases}
\end{align}
The PDF of $Y$ is obtained by differentiating \eqref{eq:CDF_Y} as
\begin{align}\label{eq:PDF_Y}
f_Y(x) 
    = \begin{cases} 
    \frac{f_D(x)}{F_D(\dth)}, & \mbox{if  } a < x \le \dth,  \\
    0, & \mbox{otherwise}.
    \end{cases}
\end{align}

Similarly, we denote by $Z$ the distance between the terminal and the serving satellite, given that the serving satellite is in $\Asl$, i.e., $\BPPml = \emptyset$ and $\BPPsl \neq \emptyset$.
With the fact that $\P[\BPPml =\emptyset, \BPPsl \neq \emptyset]=\P[\dth < D < \dmax]$, the CDF and PDF of $Z$ are respectively given by 
\begin{align}\label{eq:CDF_Z}
F_Z(x) 
    &= \P[Z \le x]\nonumber\\
    &= \P[D \le x|\dth < D < \dmax]\nonumber\\
    &= \begin{cases} 
    0, & \mbox{if  } x \le \dth,\\
    \frac{F_D(x)-F_D(\dth)}{F_D(\dmax)-F_D(\dth)}, & \mbox{if  } \dth < x \le \dmax,  \\
    1, & \mbox{if  } x > \dmax,
    \end{cases}
\end{align}
and
\begin{align}\label{eq:PDF_Z}
f_Z(x) 
    = \begin{cases} 
    \frac{f_D(x)}{F_D(\dmax)-F_D(\dth)}, & \mbox{if  } \dth < x \le \dmax,  \\
    0, & \mbox{otherwise}.
    \end{cases}
\end{align}
The derived distance distributions will be used to obtain the exact system performance in the following section.

\begin{figure*}[t]
\setcounter{myeqncount}{\value{equation}}
\setcounter{equation}{22}
\begin{align}\label{eq:Poutml_fin}
\Poutml
    &= \frac{2 K S}{(4\re(\re+a))^S F_D(\dth)}\sum\limits_{n = 0}^\infty  {\frac{{{{(m)}_n}{\delta ^n}{{(2b)}^{1 + n}}}}{{{{(n!)}^2}}}} \sum_{k=0}^{S-1}{\binom{S-1}{k}}(a+2\re )^{2(S-1-k)}(-1)^k \left(\frac{\dth^{2(k+1)}}{2(k+1)}\gamma\left(1+n,\frac{w_1\dth^{\alpha}}{2b}\right)\right. \nonumber\\
    &\qs \left. \frac{a^{2(k+1)}}{2(k+1)}\gamma\left(1+n,\frac{w_1 a^{\alpha}}{2b}\right) - \frac{(2b/w_1)^{\frac{2(k+1)}{\alpha}}}{2(k+1)}\left(\gamma\left(1+n+\frac{2(k+1)}{\alpha},\frac{w_1\dth^{\alpha}}{2b}\right)-\gamma\left(1+n+\frac{2(k+1)}{\alpha},\frac{w_1 a^{\alpha}}{2b}\right)\right)\right)
\end{align}
\normalsize \hrulefill \vspace*{4pt}
\begin{align}\label{eq:Poutsl_fin}
\Poutsl
    &= \frac{2 K S}{(4\re(\re+a))^S(F_D(\dmax)-F_D(\dth))}\sum\limits_{n = 0}^\infty  {\frac{{{{(m)}_n}{\delta ^n}{{(2b)}^{1 + n}}}}{{{{(n!)}^2}}}} \sum_{k=0}^{S-1}{\binom{S-1}{k}}(a+2\re )^{2(S-1-k)}(-1)^k \left(\frac{\dmax^{2(k+1)}}{2(k+1)}\gamma\left(1+n,\frac{w_2\dmax^{\alpha}}{2b}\right)\right. \nonumber\\
    &\qs \left. \frac{\dth^{2(k+1)}}{2(k+1)}\gamma\left(1+n,\frac{w_2 \dth^{\alpha}}{2b}\right) - \frac{(2b/w_2)^{\frac{2(k+1)}{\alpha}}}{2(k+1)}\left(\gamma\left(1+n+\frac{2(k+1)}{\alpha},\frac{w_2\dmax^{\alpha}}{2b}\right)-\gamma\left(1+n+\frac{2(k+1)}{\alpha},\frac{w_2 \dth^{\alpha}}{2b}\right)\right)\right)
\end{align}
\normalsize \hrulefill \vspace*{4pt}
\setcounter{equation}{\value{myeqncount}}
\end{figure*}

\section{Exact Performance Analyses}\label{sec:exact}
In this section, we first identify three possible distribution cases for the serving satellite and then derive the probabilities of these cases based on the characteristics of the BPP.
We also analytically derive the exact expression for the outage probability.

\subsection{Distribution Cases For Serving Satellite}\label{sec:exact:DistCases}
Since the satellites are randomly distributed over the sphere, there can be three possible cases for the serving satellite's distribution as follows:
\begin{itemize}
    \item Case 1: The serving satellite is in $\Aml$.
    \item Case 2: The serving satellite is in $\Asl$.
    \item Case 3: There is no serving satellite, i.e., all satellites are invisible.
\end{itemize}
By using the void probability of the BPP, i.e., the probability that there is no point in a certain region, the probabilities of the three cases are obtained in the following lemma.

\begin{lem}\label{lem:Prob}
The probabilities of the three distribution cases for the serving satellite are respectively given by
\begin{align}\label{eq:Pserml}
\Pserml
    &= 1 - \left(\frac{1+\cos\psth}{2}\right)^S,
\end{align}
\begin{align}\label{eq:Psersl}
\Psersl
    &= \left(\frac{1+\cos\psth}{2}\right)^S - \left(1-\frac{\dmax^2 - a^2}{4\re(\re + a)}\right)^S,
\end{align}
and
\begin{align}\label{eq:Pinv}
\Pinv
    = \left(1-\frac{\dmax^2 - a^2}{4\re(\re + a)}\right)^S.
\end{align}
\end{lem}

\begin{proof}[Proof:\nopunct]
See Appendix \ref{app:lem_prob}.
\end{proof}

\begin{rem}
It can be seen from $\Pserml$ and $\Psersl$ that as the beamwidth of the satellites' main lobes increases, the terminal is more likely to be associated with the satellite whose main lobe is directed to the terminal. This is because $\psth$ increases with the beamwidth of the satellites.
\end{rem}
\begin{rem}
The satellite-visible probability, i.e., the probability that there exist at least one visible satellite, is $\Pvis=1-\Pinv$. As the minimum elevation angle $\themin$ increases, $\Pvis$ decreases, while $\Pinv$ increases, because the surface area of $\Avis$ shrinks.
\end{rem}

\subsection{Outage Probability}
In this subsection, we analyze the outage probability of the system, assuming that at least one satellite is visible.
An outage occurs when the instantaneous rate between the terminal and the serving satellite falls below a required transmission rate~$R$.
The outage probability of the system is obtained in the following theorem.

\begin{thm}\label{thm:Pout}
The outage probability of the system is given by
\begin{align}\label{eq:P_out_fin}
\Pout
    &= \frac{\Pserml}{1-\Pinv} \Poutml + \frac{\Psersl}{1-\Pinv} \Poutsl
\end{align}
where $\Poutml$ and $\Poutsl$ are the outage probabilities when the serving satellite is in $\Aml$ or $\Asl$, respectively given by \eqref{eq:Poutml_fin} and \eqref{eq:Poutsl_fin} shown at the top of this page with $w_1 = (16 \pi^2 f_{\mathrm{c}}^2 N_0 W(2^{R}-1))/(P g c^2 \Gtml \Gr)$ and $w_2 = (16 \pi^2 f_{\mathrm{c}}^2 N_0 W(2^{R}-1))/(P g c^2 \Gtsl \Gr)$.
\end{thm}
\setcounter{equation}{24}

\begin{proof}[Proof:\nopunct]
When the terminal is associated with the serving satellite $\s0$, the outage probability of the system is given by
\begin{align}\label{eq:Pout1}
\Pout 
    &= \P[\s0 \in \BPPml|\BPPvis \neq \emptyset] \Poutml + \P[\s0 \in \BPPsl|\BPPvis \neq \emptyset] \Poutsl
\end{align}
where $\P[\s0 \in \BPPvis^{i}|\BPPvis \neq \emptyset]$, $i\in\{\mathrm{ml},\mathrm{sl}\}$ is given by
\begin{align}
\frac{\P[\s0 \in \BPPvis^{i}, \BPPvis \neq \emptyset]}{\P[\BPPvis \neq \emptyset]} = \frac{P_{\mathrm{ser}}^{i}}{1-\Pinv}.
\end{align}
The outage probability of the system given $\s0 \in \BPPml$, $\Poutml$, in \eqref{eq:Pout1} is derived as
\begin{align}\label{eq:Pout2}
\Poutml
    &= \P[\log(1+\gamma_{\s0}) < R|\s0 \in \BPPml]\nonumber\\
    &=  \P\left[h_{\s0}  < \frac{16 \pi^2 f_{\mathrm{c}}^2 N_0 W(2^{R}-1)}{P g c^2 \Gtml \Gr d_{\s0}^{-\alpha}}\right]\nonumber\\
    & = \int_{a}^{\dth} F_{h_{\s0}}(w_1 x^{\alpha}) f_Y(x)dx.
\end{align}
Substituting \eqref{eq:CDF_ch_gain} and \eqref{eq:PDF_Y} into \eqref{eq:Pout2}, we have
\begin{align}\label{eq:Poutml1}
\Poutml
    &= \frac{K S}{2\re(\re+a)F_D(\dth)}\sum\limits_{n = 0}^\infty  {\frac{{{{(m)}_n}{\delta ^n}{{(2b)}^{1 + n}}}}{{{{(n!)}^2}}}} \nonumber\\
    &\qt  \int_{a}^{\dth} \gamma\left(1+n,\frac{w_1 x^{\alpha}}{2b}\right) \left(1-\frac{x^2-a^2}{4\re(\re+a)}\right)^{S-1} x dx \nonumber\\
    &\mathop=^{(a)} \frac{2 K S}{(4\re(\re+a))^S F_D(\dth)}\sum\limits_{n = 0}^\infty  {\frac{{{{(m)}_n}{\delta ^n}{{(2b)}^{1 + n}}}}{{{{(n!)}^2}}}} \nonumber\\
    &\qt  \sum_{k=0}^{S-1}{\binom{S-1}{k}}(a+2\re )^{2(S-1-k)}(-1)^k \nonumber\\
    &\qt \int_{a}^{\dth} \int_0^{\frac{w_1 x^{\alpha}}{2b}} t^n e^{-t} x^{2k+1} dt  dx
\end{align}
where ($a$) follows from the binomial expansion and the definition of the lower incomplete Gamma function.

\begin{figure}
\centering
\includegraphics[width=.8\columnwidth]{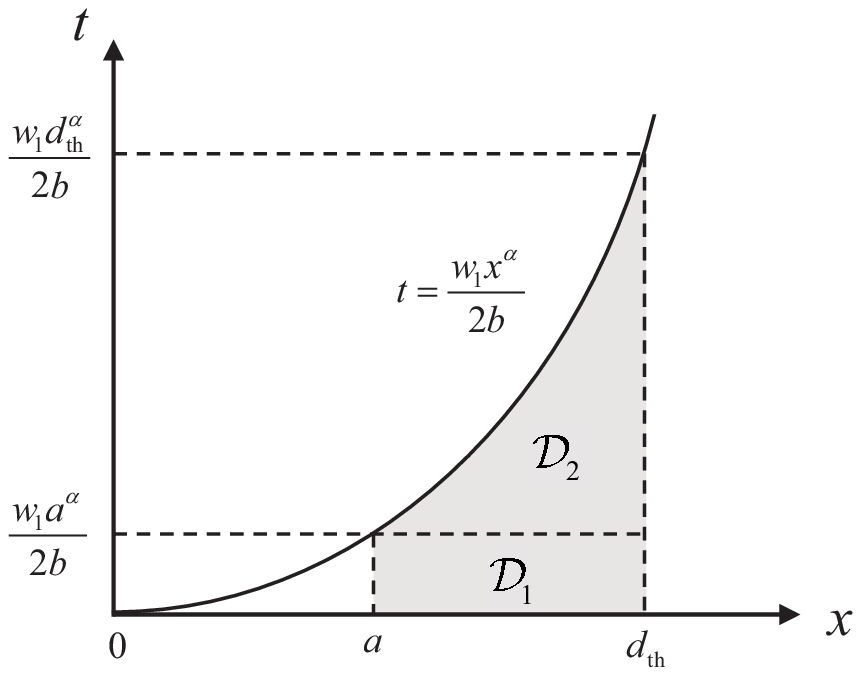}
\caption{Domain of the integration in \eqref{eq:Poutml1}.}
\label{Fig:Region1}
\end{figure}

The domain of the integration in (\ref{eq:Poutml1}) is shown as the shaded areas in Fig. \ref{Fig:Region1} and can be divided into two domains $\mathcal{D}_1$ and $\mathcal{D}_2$, which are respectively given by
$
    \mathcal{D}_1 = \left\{0 \le t \le \frac{ w_1 a^\alpha}{2 b}, a \le x \le \dth \right\}\nonumber
$
and
$
    \mathcal{D}_2 =  \left\{\frac{w_1 a^\alpha}{2 b} \le t \le \frac{w_1 \dth^\alpha}{2 b}, \left(\frac{2 b t}{w_1}\right)^{\frac{1}{\alpha}} \le x \le \dth \right\}.\nonumber
$
In order to calculate the integral in (\ref{eq:Poutml1}), we convert the double integral into the sum of two integrals over the two domains $\mathcal{D}_1$ and $\mathcal{D}_2$, respectively.
The integral over the domain $\mathcal{D}_1$ is a product of two integrals over $x$ and $t$, which is given by
\begin{align}\label{eq:integral_s1}
\mathcal{I}_{\mathcal{D}_1}
    &= \int_{a}^{\dth} x^{2k+1} dx  \times  \int_0^{\frac{w_1 a^{\alpha}}{2b}}   t^n e^{-t} dt\nonumber\\
    &= \frac{\dth^{2(k+1)}-a^{2(k+1)}}{2(k+1)} \gamma\left(1+n, \frac{w_1 a^\alpha}{2 b}\right).
\end{align}
By changing the order of variables, the integral over the domain $\mathcal{D}_2$ is given by
\begin{align}\label{eq:integral_s2}
\mathcal{I}_{\mathcal{D}_2}
    &= \int_{\frac{w_1 a^\alpha}{2 b}}^{\frac{w_1 \dth^\alpha}{2 b}} \int_{\left(\frac{2 b t}{w_1}\right)^{\frac{1}{\alpha}}}^{\dth} t^n e^{-t}  x^{2k+1} dx dt \nonumber\\
    &= \frac{1}{2(k+1)} \int_{\frac{w_1 a^\alpha}{2 b}}^{\frac{w_1 \dth^\alpha}{2 b }} t^n \exp(-t) \left\{ \dth^{2(k+1)}-\left(\frac{2 b t}{w_1}\right)^{\frac{2(k+1)}{\alpha}}\right\} dt \nonumber\\
    &= \frac{\dth^{2(k+1)}}{2(k+1)} \left\{\gamma\left(1+n,\frac{w_1 \dth^\alpha}{2 b}\right) - \gamma\left(1+n,\frac{w_1 a^\alpha}{2 b }\right)\right\} \nonumber\\
    & \qs \frac{(2 b/w_1)^{\frac{2(k+1)}{\alpha}}}{2(k+1)} \left\{\gamma\!\left(1+n+\!\frac{2(k\!+\!1)}{\alpha},\frac{w_1 \dth^\alpha}{2 b} \right) \right. \nonumber\\
    &\qs \left. \gamma\!\left(1+n+\!\frac{2(k\!+\!1)}{\alpha},\frac{w_1 a^\alpha}{2 b}\right)\right\}.
\end{align}
From (\ref{eq:Poutml1})-(\ref{eq:integral_s2}), the final expression of $\Poutml$ is given in \eqref{eq:Poutml_fin}. With the similar steps to derive $\Poutml$, the final expression of $\Poutsl$ can be easily obtained as \eqref{eq:Poutsl_fin}.
From \eqref{eq:Poutml_fin} and \eqref{eq:Poutsl_fin} with the results in Lemma~\ref{lem:Prob}, the final expression $\Pout$ can be obtained.
\end{proof}

\begin{figure}
\includegraphics[width=0.97\columnwidth]{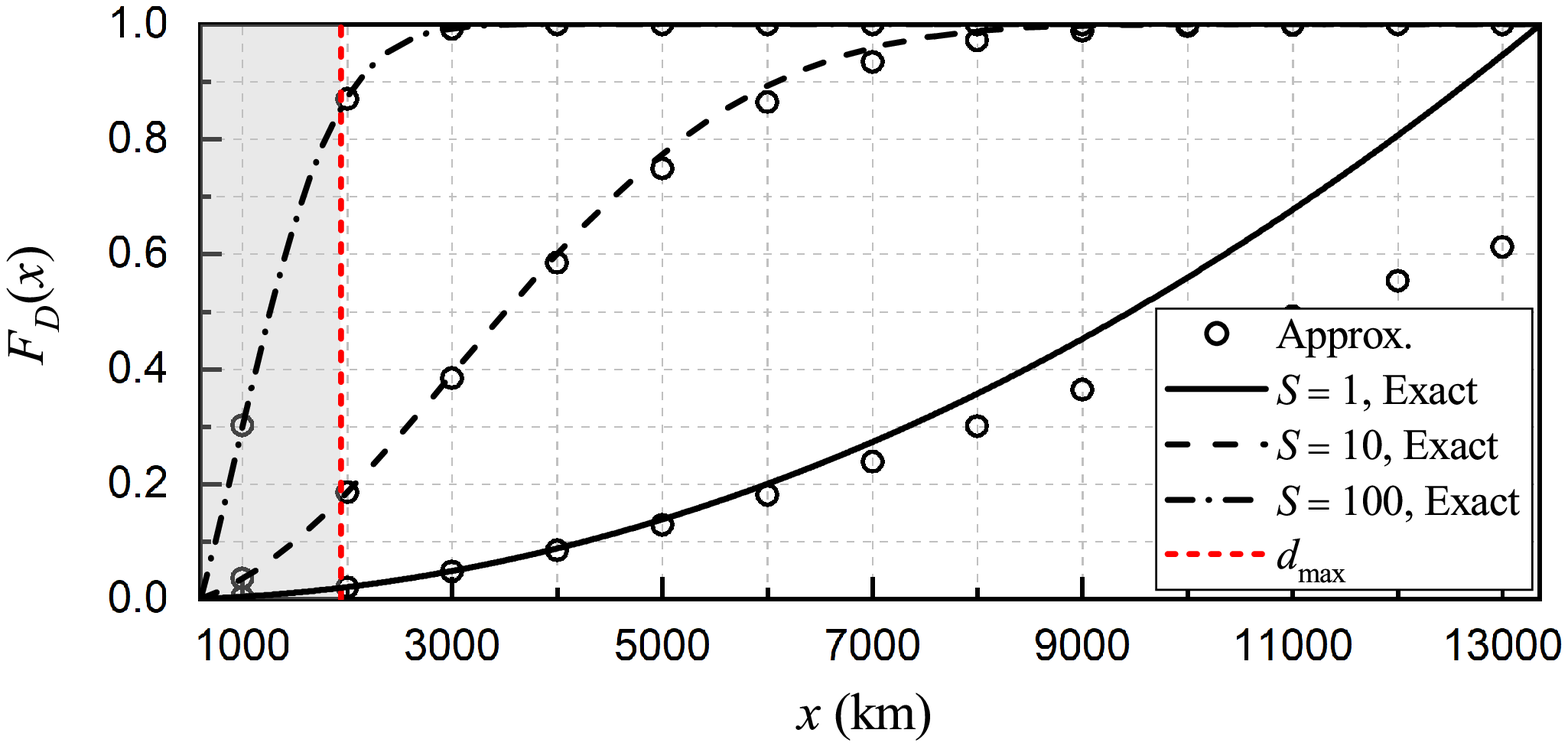}
\hspace{2cm}%
\includegraphics[width=0.97\columnwidth]{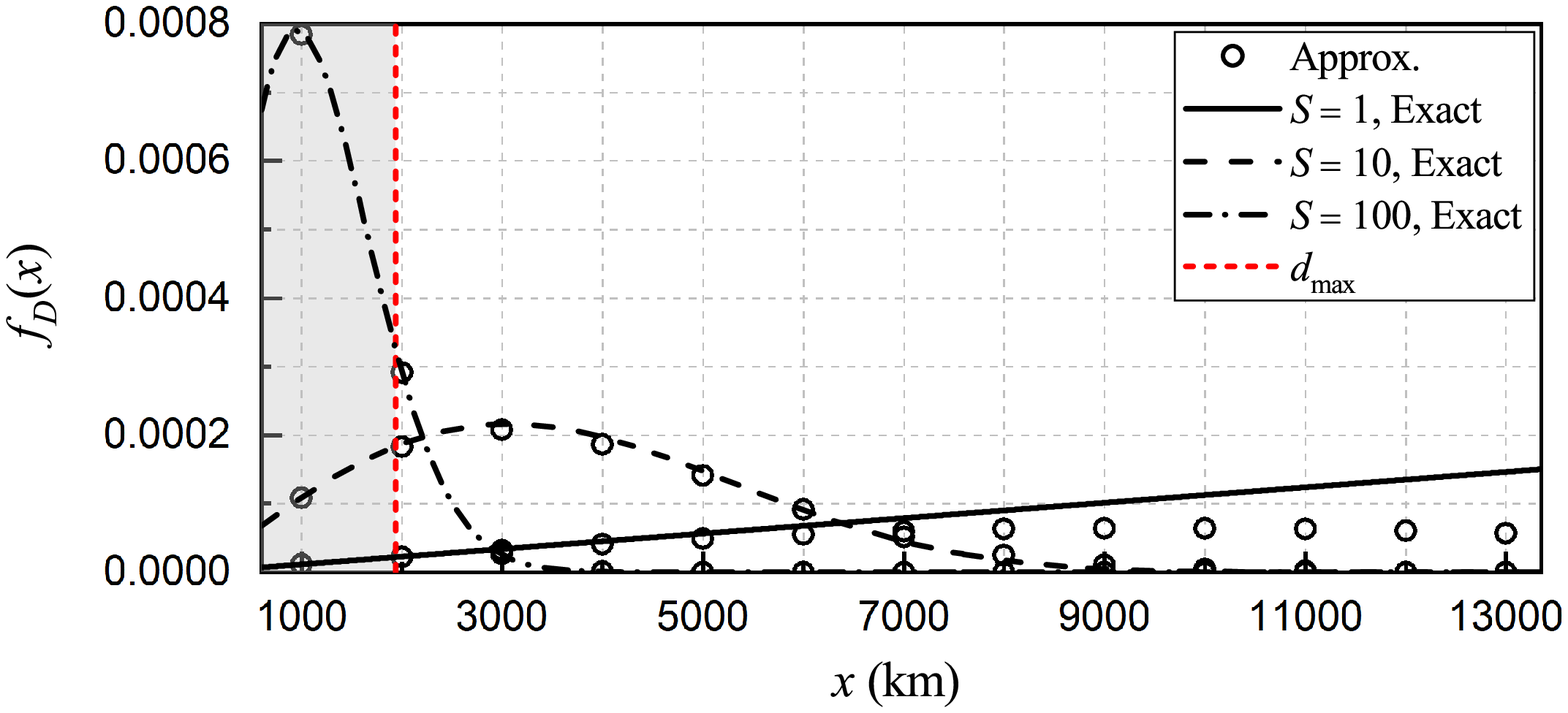}
\caption{CDF and PDF of distance to the nearest satellite, $F_D(x)$ and $f_D(x)$, for various altitudes $S=\{1,10,100\}$ km. $a=600$ km.}
\label{Fig:CDFPDFD}
\end{figure}

The expression of the outage probability is exact but complicated to obtain any insight for the system performance. In Section \ref{sec:app}, we obtain much simpler and tight approximated expression for the outage probability.

\section{Performance Approximation}\label{sec:app}
In this section, we first derive approximated expressions for the probabilities of the three distribution cases for the serving satellite and the outage probability using the Poisson limit theorem. Next, we obtain an asymptotic expression as the number of the satellites goes to infinity, which can be applicable to ultra dense LEO satellite scenarios. We also analyze the convergence of the approximated outage probability. 

We assume that the altitude of the satellites is sufficiently low, e.g.,  LEO  and  very  low  earth  orbit satellites. 
Then, based on the Poisson limit theorem, the satellites in a bounded area are asymptotically distributed according to a PPP $\PPP$ whose density is given by $\lambda_{\mathrm{s}}=S/(4\pi(\re +a)^2)$ [\ref{Ref:Book:Chiu}].
Since the void probability of the PPP in a region $\mathcal{X}$ is given by $\PvoidPPP=e^{-\lambda_{\mathrm{s}} S(\mathcal{X})}$ [\ref{Ref:Andrews}],
the CDF of distance between the terminal and the nearest satellite can be approximated as
\begin{align}\label{eq:CDF_D_approx}
\bar{F}_D^{\mathrm{}}(x)
    &=  1 - \P[D>x]\nonumber\\
    &=  1 - \prod_{s \in \PPP} \P[d_s>x]\nonumber\\
    &=  1- e^{-\lambda_{\mathrm{s}} \mathcal{S}(\mathcal{A}(x))}\nonumber\\ 
    &=
    \begin{cases} 
     0, & \mbox{if } x \le a,\\
     1- e^{-\frac{S(x^2-a^2)}{4 \re (\re +a)}}, & \mbox{if } a < x \le 2\re+a, \\
     1, & \mbox{if } x > 2\re+a,
     \end{cases}
\end{align}
and the corresponding PDF is given by 
\begin{align}\label{eq:PDF_D_approx}
\bar{f}_D^{\mathrm{}}(x)
    &=  
    \begin{cases} 
     \frac{S x}{2 \re (\re +a)}e^{-\frac{S(x^2-a^2)}{4 \re (\re +a)}}, & \mbox{if } a < x \le 2\re+a,\\
     0, & \mbox{otherwise}.
     \end{cases}
\end{align}

\begin{figure}
\includegraphics[width=\columnwidth]{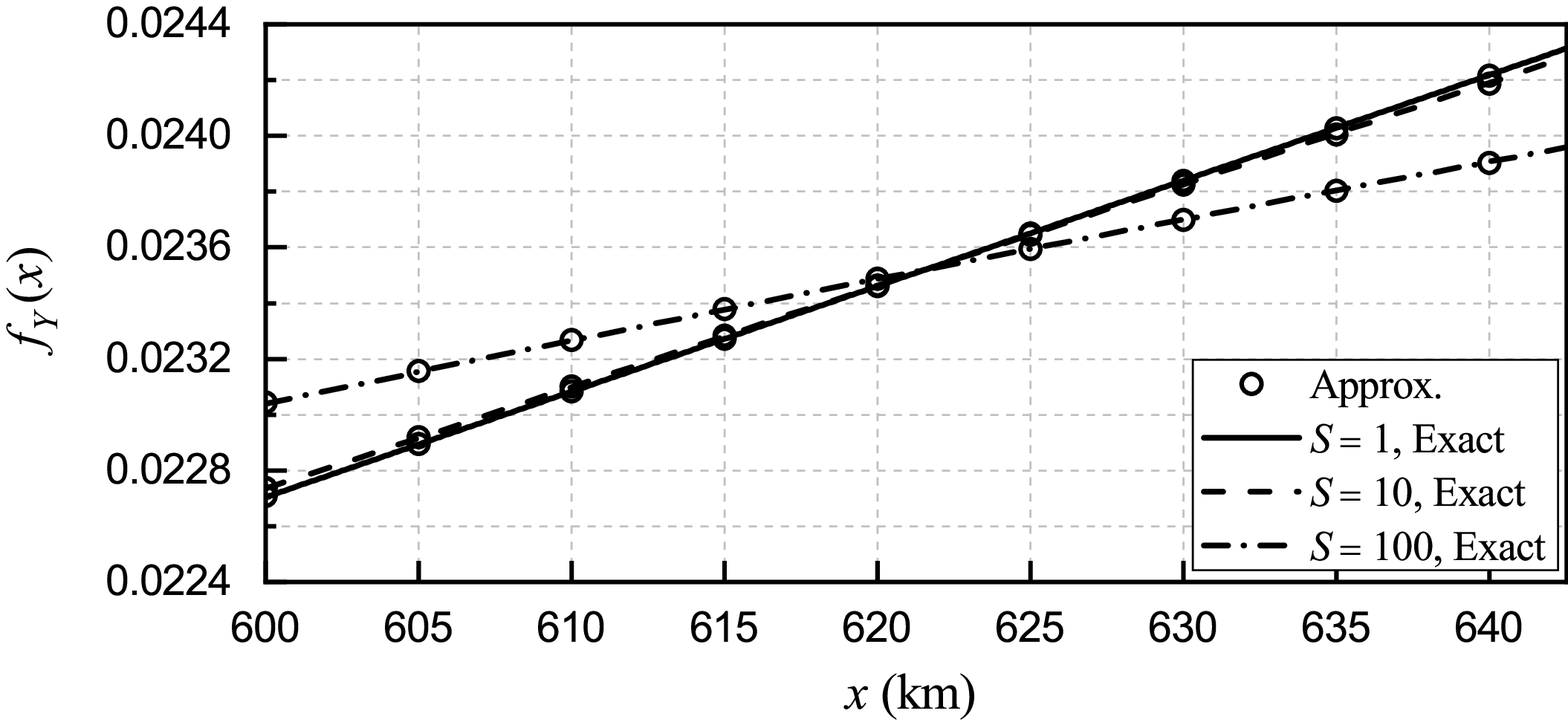}
\hspace{2cm}%
\includegraphics[width=\columnwidth]{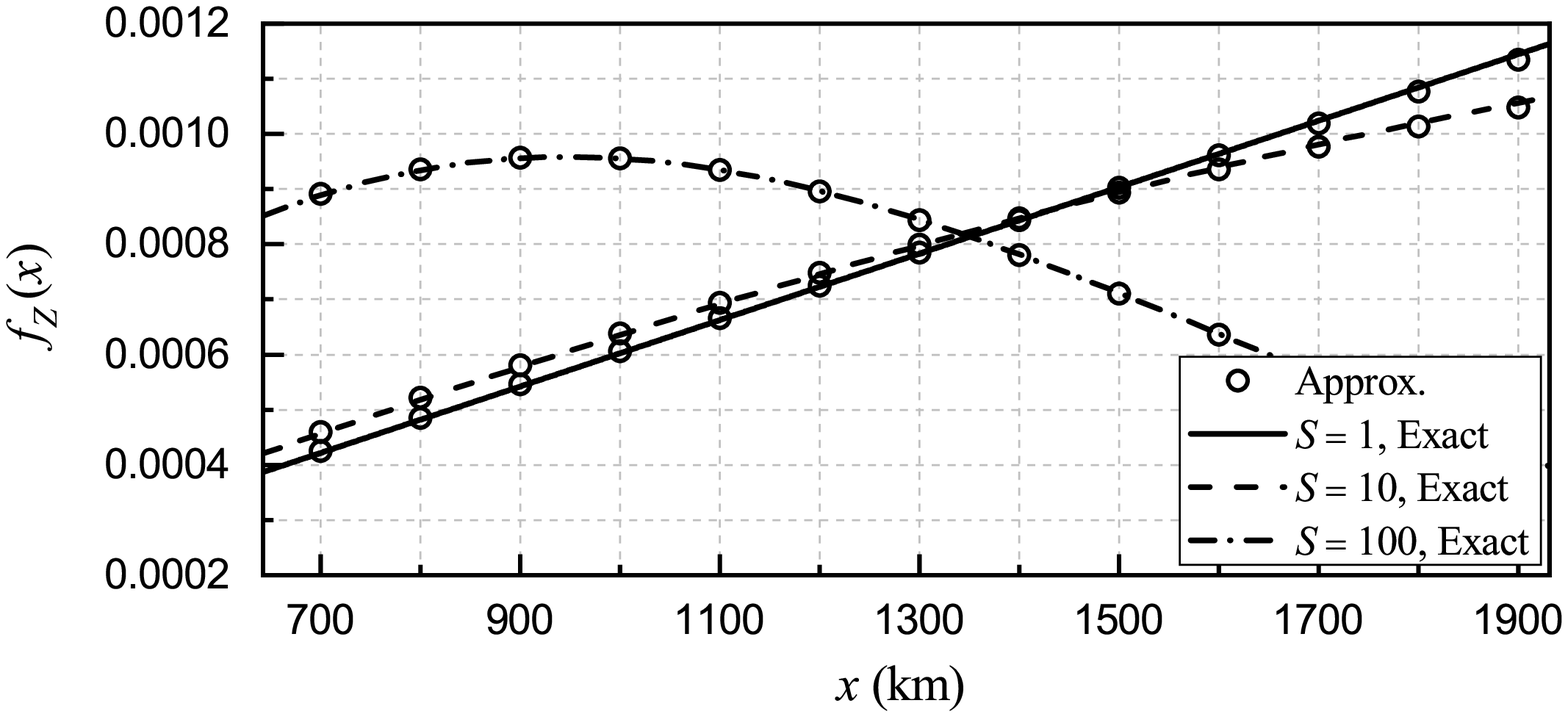}
\caption{PDFs of $Y$ and $Z$, $f_Y(x)$ and $f_Z(x)$, for various altitudes $S=\{1,10,100\}$ km. $a=600$ km.}
\label{Fig:PDFYZ}
\end{figure}

The approximated CDF and PDF of $D$ are compared with the exact ones in Fig. \ref{Fig:CDFPDFD}. 
The shaded area is the satellite-visible region including the distance between the terminal and the visible satellites, while the other area is the satellite-invisible region that is not considered for the performance analyses.
As $S$ increases, the nearest satellite is more likely to be located close, because of the satellites' dense distribution.
The approximated results become closer to the exact ones as $x$ decreases.
Especially for the satellite-visible region, both results are almost the same.
This verify that the Poisson limit theorem is well-applicable to approximate the BPP as the PPP for the LEO satellites' distribution.

The approximated CDFs and PDFs of $Y$ and $Z$, denoted by $\bar{F}_Y(x)$, $\bar{F}_Z(x)$, $\bar{f}_Y(x)$, and $\bar{f}_Z(x)$, can be obtained by substituting \eqref{eq:CDF_D_approx} and \eqref{eq:PDF_D_approx} into \eqref{eq:CDF_Y}-\eqref{eq:PDF_Z}.
As expected, the exact and approximated PDFs of $Y$ and $Z$ are also fairly close as shown in Fig. \ref{Fig:PDFYZ}. It is shown that as $S$ increases, the distance to the serving satellite is more likely to be closer. The approximated probabilities of the distribution cases for the serving satellite can be obtained in the following lemma.

\begin{lem}\label{lem:prob_approx}
The probabilities of the three distribution cases for the serving satellite are respectively approximated as
\begin{align} \label{eq:Psermlapp}
\Psermlapp
    &= 1- e^{-\frac{S}{2}(1-\cos\psth)},
\end{align}
\begin{align}\label{eq:Pserslapp}
\Pserslapp
    &= e^{-\frac{S}{2}(1-\cos\psth)} - e^{-\frac{S(\dmax^2-a^2)}{4\re(\re+a)}},
\end{align}
and
\begin{align}
\Pinvapp
    = e^{-\frac{S(\dmax^2-a^2)}{4\re(\re+a)}}.
\end{align}
\end{lem}

\begin{proof}[Proof:\nopunct]
The proof is similar to that of Lemma \ref{lem:Prob} with the void probability of the PPP.
\end{proof}

Using Lemma \ref{lem:prob_approx}, the approximated outage probability is obtained in the following theorem.

\begin{thm}\label{thm:Poutapprox}
The outage probability of the system is approximated as
\begin{align}\label{eq:Poutapp_fin}
\Poutapp
    &= \frac{\Psermlapp}{1-\Pinvapp} \Poutmlapp + \frac{\Pserslapp}{1-\Pinvapp} \Poutslapp
\end{align}
where $\Poutmlapp$ and $\Poutslapp$ are the approximated outage probabilities when the serving satellite is in $\Aml$ or $\Asl$, respectively given by
\begin{align}\label{eq:Poutml_app_fin}
\Poutmlapp
    & = \frac{K}{\bar{F}_D(\dth)}\sum\limits_{n = 0}^\infty  {\frac{{{{(m)}_n}{\delta ^n}{{(2b)}^{1 + n}}}}{{{{(n!)}^2}}}} \left(e^{\frac{S a^2}{4\re(\re+a)}} \mathcal{C}_\mathrm{ml}[n]\right.\nonumber\\
    &\qa \gamma\left(1+n,\frac{w_1 a^{\alpha}}{2b}\right) - \left.e^{-\kappa(\dth)}\gamma\left(1+n, \frac{w_1 \dth^{\alpha}}{2b}\right) \right)
\end{align}
and 
\begin{align}\label{eq:Poutsl_app_fin}
\Poutslapp
    & = \!\frac{K}{\bar{F}_D(\dmax)\!-\!\bar{F}_D(\dth)}\!\sum\limits_{n = 0}^\infty  \!{\frac{{{{(m)}_n}{\delta ^n}{{(2b)}^{1 + n}}}}{{{{(n!)}^2}}}}\!\left(e^{\frac{S a^2}{4\re(\re+a)}} \mathcal{C}_\mathrm{sl}[n]\right.\nonumber\\
    & + \left. e^{-\kappa(\dth)}\gamma\left(1+n, \frac{w_2 \dth^{\alpha}}{2b}\right) - e^{-\kappa(\dmax)}\gamma\left(1+n, \frac{w_2 \dmax^\alpha}{2b}\right)\right)
\end{align}
with 
\begin{align}\label{eq:Cml}
\mathcal{C}_\mathrm{ml}[n]
    = \int_{\frac{w_1 a^{\alpha}}{2b}}^{\frac{w_1 \dth^{\alpha}}{2b}} t^n e^{-t-\frac{S}{4\re(\re+a)}\left(\frac{2 b t}{w_1}\right)^{\frac{2}{\alpha}}} dt
\end{align}
and
\begin{align}\label{eq:Csl}
\mathcal{C}_\mathrm{sl}[n]
    = \int_{\frac{w_2 \dth^{\alpha}}{2b}}^{\frac{w_2 \dmax^{\alpha}}{2b}} t^n e^{-t-\frac{S}{4\re(\re+a)}\left(\frac{2 b t}{w_2}\right)^{\frac{2}{\alpha}}} dt.
\end{align}
\end{thm}

\begin{proof}[Proof:\nopunct]
See Appendix \ref{app:thm_pout_approx}.
\end{proof}

\begin{rem}
The approximated outage probability in Theorem~\ref{thm:Poutapprox} is much simpler than the exact one in Theorem \ref{thm:Pout}, since the approximated expression have no summation over the number of satellites thanks to the Poisson limit theorem.
Especially, when there are an extensive number of satellites, e.g., thousands of satellites, the approximated results make it easy to evaluate the system performance. 
\end{rem}

The numerical integrals in \eqref{eq:Cml} and \eqref{eq:Csl} can be further simplified for $\alpha=2$ in the following corollary.

\begin{cor}\label{cor:Poutapprox}
When $\alpha=2$, $\mathcal{C}_\mathrm{ml}[n]$ and $\mathcal{C}_\mathrm{ml}[n]$ become
\begin{align}\label{eq:Cmla2}
\mathcal{C}_\mathrm{ml}&[n,\alpha=2]\nonumber\\
    &= \frac{1}{w_3^{1+n}}\left(\gamma\left(1+n,\frac{w_1 w_3 \dth^{2}}{2b}\right)-\gamma\left(1+n,\frac{w_1 w_3 a^{2}}{2b}\right)\right)
\end{align}
and
\begin{align}\label{eq:Csla2}
&\mathcal{C}_\mathrm{sl}[n,\alpha=2]\nonumber\\
    &= \frac{1}{w_4^{1+n}}\left(\gamma\left(1+n,\frac{w_2 w_4 \dmax^{2}}{2b}\right)-\gamma\left(1+n,\frac{w_2 w_4 \dth^{2}}{2b}\right)\right),
\end{align}
respectively, where $w_3=1+2 S b/(4w_1 \re (\re+a))$ and $w_4=1+2 S b/(4w_2 \re (\re+a))$.
\end{cor}

\begin{proof}[Proof:\nopunct]
By letting $\alpha=2$ and using the change of variable $l=w_3 t$, we have
\begin{align}
\mathcal{C}_\mathrm{ml}[n,\alpha=2]
    = \int_{\frac{w_1 a^{2}}{2b}}^{\frac{w_1 \dth^{2}}{2b}} t^n
    e^{-w_3t} dt
    = \frac{1}{w_3^{1+n}}\int_{\frac{w_1 w_3 a^{2}}{2b}}^{\frac{w_1 w_3 \dth^{2}}{2b}} l^n e^{-l} dl,
\end{align}
which becomes \eqref{eq:Cmla2} from the definition of the lower incomplete Gamma function.
Similarly, $\mathcal{C}_\mathrm{sl}[n, \alpha=2]$ in \eqref{eq:Csla2} can be readily obtained.
\end{proof}

We also conducted the asymptotic analysis for $S\to\infty$ with arbitrary $\alpha$ in the following corollary.

\begin{cor}\label{cor:PoutapproxAsym}
When $S\to\infty$, $\Poutapp$ becomes
\begin{align}\label{eq:PoutapproxAsym}
\Poutapp 
    \to K\sum\limits_{n = 0}^\infty  {\frac{{{{(m)}_n}{\delta ^n}{{(2b)}^{1 + n}}}}{{{{(n!)}^2}}}}\gamma\left(1+n,\frac{w_1 a^{\alpha}}{2b}\right).
\end{align}
\end{cor}

\begin{proof}[Proof:\nopunct]
Assuming that $S\to\infty$, $\Psermlapp\to1$ and $\Pserslapp$, $\Pinvapp\to~0$ from the results in Lemma \ref{lem:prob_approx}.  
In addition, the approximated CDF and PDF of $D$ in \eqref{eq:CDF_D_approx} and \eqref{eq:PDF_D_approx} asymptotically become  $\bar{F}_D(x)\to u(x-a)$ and $\bar{f}_D(x)\to \delta(x-a)$, respectively, where $u(x)$ is the unit step function, and $\delta(x)$ is the Dirac delta function. From \eqref{eq:CDF_Y} and \eqref{eq:PDF_Y}, the approximated CDF and PDF of $Y$ also become $u(x-a)$ and $\delta(x-a)$, respectively. Using the asymptotic $\Psermlapp$, $\Pserslapp$, $\Pinvapp$, and  $\bar{f}_Y(x)$, 
the outage probability is asymptotically obtained as $\Poutapp=\Poutmlapp=F_{h_{s_0}}(w_1 a^{\alpha})$, which becomes \eqref{eq:PoutapproxAsym} by \eqref{eq:CDF_ch_gain}.
\end{proof}

\begin{rem}
When $S\to\infty$, the nearest satellite is surely located in $\Aml$ due to the large number of satellites, i.e., $\Psermlapp\to1$  and $\Pserslapp,\Pinvapp \to 0$. 
\end{rem}
\begin{rem}
When $S\to\infty$, $\Poutapp$ becomes the true outage probability when the distance to the serving satellite is deterministic and has value of $a$, i.e., possible minimum distance to the satellite.
\end{rem}

Since the expression of the outage probability in Theorem~\ref{thm:Poutapprox} has the infinite number of summations, we now analyze its convergence.
Let $\Poutapp[N]$ be the approximated outage probability where the infinite summations in \eqref{eq:Poutml_app_fin} and \eqref{eq:Poutsl_app_fin} are limited to the summations over $n=0,\cdots,N$.
Then, the difference between $\Poutapp[N]$ and $\Poutapp[N-1]$ is given by
\begin{align}\label{eq:PoutDelta}
\Delta_{\mathrm{out}}[N]
    &\delequal\Poutapp[N] - \Poutapp[N-1]\nonumber\\
    &= \frac{\Pserml}{1-\Pinv} \Delta_{\mathrm{out}}^{\mathrm{ml}}[N] + \frac{\Psersl}{1-\Pinv} \Delta_{\mathrm{out}}^{\mathrm{sl}}[N]
\end{align}
where
\begin{align}\label{eq:deltaml}
\Delta_{\mathrm{out}}^{\mathrm{ml}}&[N]
     = \frac{2b K{(m)}_N (2b\delta)^N}{\bar{F}_D(\dth)(N!)^2} \left(e^{\frac{S a^2}{4\re(\re+a)}} \mathcal{C}_\mathrm{ml}[N]\right.\nonumber\\
    &\qa \gamma\left(1+N,\frac{w_1 a^{\alpha}}{2b}\right) - \left.e^{-\kappa(\dth)}\gamma\left(1+N, \frac{w_1 \dth^{\alpha}}{2b}\right) \right)
\end{align}
and
\begin{align}\label{eq:deltasl}
&\Delta_{\mathrm{out}}^{\mathrm{sl}}[N]
     = \frac{2b K {(m)}_N (2b\delta)^N}{(\bar{F}_D(\dmax)-\bar{F}_D(\dth)){(N!)}^2}\left(e^{\frac{S a^2}{4\re(\re+a)}} \mathcal{C}_\mathrm{sl}[N]\right.\nonumber\\
    & + \left. e^{-\kappa(\dth)}\gamma\left(1+N, \frac{w_2 \dth^{\alpha}}{2b}\right) - e^{-\kappa(\dmax)}\gamma\left(1+N, \frac{w_2 \dmax^\alpha}{2b}\right)\right).
\end{align}
In \eqref{eq:deltaml}, 
\begin{align}\label{eq:convMulti}
\lim_{N\to\infty} \frac{(m)_N}{(N!)^2} = \lim_{N\to\infty}\frac{\Gamma(m+N)}{\Gamma(m)\Gamma^2(N+1)} = 0,    
\end{align}
so $\Delta_{\mathrm{out}}^{\mathrm{ml}}[N]$ converges when the three terms $\Cml[N]/\zeta^N$, $\gamma\left(1+N,\frac{w_1 a^{\alpha}}{2b}\right)/\zeta^N$, and $\gamma\left(1+N, \frac{w_1 \dth^{\alpha}}{2b}\right)/\zeta^N$ converge where $\zeta=1/2b\delta$.
Let $\mcV_N^{(i)}=\Cml[N]/\zeta^{N+i}$. Then, as $N\to\infty$, $\mcV_N^{(i)}$ becomes
\begin{align}\label{eq:limitV}
\lim_{N\to \infty}\mcV_{N}^{(i)} 
    &= \lim_{N\to \infty} \frac{\Cml[N]}{\zeta^{N+i}} \nonumber\\
    &= \beta \lim_{N\to \infty}  \frac{\Cml[N]-\Cml[N-1]}{\zeta^{N+i}-\zeta^{N+i-1}} + \lim_{N\to \infty} \frac{\Cml[N-1]}{\zeta^{N+i}}\nonumber\\
    &\mathop=^{(a)} \beta \lim_{N\to \infty} \mcV_{N-1}^{(i)} + \lim_{N\to \infty} \mcV_{N-1}^{(i+1)}
\end{align}
where $\beta=1-1/\zeta=1-2b\delta$, and ($a$) follows from the Stolz-Ces\`{a}ro theorem, i.e., $\lim_{n\to\infty}\frac{a_{n+1}-a_{n}}{b_{n+1}-b_{n}}\to\lim_{n\to\infty}\frac{a_{n}}{b_{n}}$ for strictly increasing sequences $a_n$ and $b_n$ [\ref{Ref:Book:Choudary}].
Using \eqref{eq:limitV}, the limitation of the first term $\Cml[N]/\zeta^N$, i.e., $\lim_{N\to \infty} \mcV_N^{(0)}$, is obtained as
\begin{align}\label{eq:conv1st}
\lim_{N\to \infty} \mcV_N^{(0)} 
    &= \beta \lim_{N\to \infty} \mcV_{N-1}^{(0)} + \lim_{N\to \infty}\mcV_{N-1}^{(1)} \nonumber\\
    &= \beta^2 \lim_{N\to \infty}\mcV_{N-2}^{(0)} + 2\beta \lim_{N\to \infty}\mcV_{N-2}^{(1)} + \lim_{N\to \infty}\mcV_{N-2}^{(2)} \nonumber\\
    &\,\,\vdots \nonumber\\
    &= \lim_{N\to \infty} \sum_{n=0}^{N} \binom{N}{n} \beta^{N-n} \mcV_{0}^{(n)}\nonumber\\
    &= \Cml[0] \lim_{N\to \infty}(\beta+2b\delta)^N \nonumber\\
    &\mathop=^{{(a)}} \Cml[0]
\end{align}
where ($a$) follows from the fact that $\lim_{N\to \infty}1^N=1$. Since $\Cml[0]$ is a constant on $N$, the first term converges as $N\to\infty$.
Similarly, the second and third terms respectively converge as 
\begin{align}\label{eq:conv2nd}
\lim_{N\to\infty} \frac{1}{\zeta^N}\gamma\left(1+N,\frac{w_1 a^{\alpha}}{2b}\right) = 1-e^{\frac{w_1 a^{\alpha}}{2b}}
\end{align}
and
\begin{align}\label{eq:conv3rd}
\lim_{N\to\infty} \frac{1}{\zeta^N} \gamma\left(1+N, \frac{w_1 \dth^{\alpha}}{2b}\right) = 1-e^{\frac{w_1 \dth^{\alpha}}{2b}}.
\end{align}
From \eqref{eq:deltaml}, \eqref{eq:convMulti} and \eqref{eq:conv1st}-\eqref{eq:conv3rd}, $\lim_{N\to\infty} \Delta_{\mathrm{out}}^{\mathrm{ml}}[N] = 0$. With the similar steps, we can show that $\lim_{N\to\infty} \Delta_{\mathrm{out}}^{\mathrm{sl}}[N] = 0$. Hence, we finally have $\lim_{N\to\infty} \Delta_{\mathrm{out}}[N] = 0$, which proves that the approximated outage probability converges as $N\to\infty$.

\begin{figure*}
\setcounter{myeqncount}{\value{equation}}
\setcounter{equation}{55}
\begin{align}\label{eq:1st_der}
\frac{d\themin}{d\dmax} 
    =& - \frac{{{a^2} + 2{\re }a + d_{{\mathrm{max}}}^2}}{{{d_{{\mathrm{max}}}}\sqrt {(2\re +\dmax+a)(\dmax-a)(2\re d_{{\mathrm{max}}}+{{a^2} + 2{\re }a - d_{{\mathrm{max}}}^2})} }}<0
\end{align}
\normalsize \hrulefill \vspace*{4pt}
\setcounter{equation}{\value{myeqncount}}
\end{figure*}

\section{Throughput Maximization and Complexity Analyses}\label{sec:thruMax}
In this section, we formulate the throughput maximization problem with satellite-visibility and outage constraints.To solve the problem, we reformulate it with bounded feasible sets and propose an iterative algorithm. We also analyze the computational complexity of the proposed algorithm.

The system throughput is defined as the data rate (bps/Hz) successfully transferred from the serving satellite to the terminal without any outage, which is given by [\ref{Ref:Nasir}]
\begin{equation}\label{eq:Thru}
    T=\Pvis(1-\Pout)R.
\end{equation}
Now, we optimize the transmission rate $R$ and the minimum elevation angel $\themin$ to maximize the system throughput under the satellite-visibility and outage constraints.\footnote{The transmission rate $R$ should be carefully selected for throughput maximization. Too low $R$ reduces the transmitted data rate itself, and too high $R$ causes the high outage probability, both resulting in the reduced system throughput. It is also crucial to configure a proper minimum elevation angle $\themin$. This is because, with low $\themin$, the distance to the serving satellite can be too long, while, with high $\themin$, the satellite-invisible probability increases, which may decrease the system throughput.}
The throughput maximization problem is formulated as
\begin{subequations}\label{prob1}
\begin{align}
{\mathop {{\mathop{\mathrm {maximize}}\nolimits} }\limits_{{R, \themin}} }&\,\,\,\,\,T(R,\themin)\label{prob_obj_func}\\
{{\mathrm{subject \,\,to}}}&\,\,\,\,\,{\Pvis(\themin) \ge {\eta}},\label{prob_const_a}\\ 
{}&\,\,\,\,\,{{\Pout(R,\themin)} \le \varepsilon,}\label{prob_const_b}\\
{}&\,\,\,\,\,{{R} \ge {0},}\label{prob_const_c}\\ 
{}&\,\,\,\,\,{{\themin} \ge {0}.\label{prob_const_d}}
\end{align}
\end{subequations}
In the problem (\ref{prob1}), the constraint (\ref{prob_const_a}) means that the satellite-visible probability for a given location is greater than a visibility constraint $\eta$ to achieve high availability, and the constraint (\ref{prob_const_b}) describes that the outage probability of the system is less than an outage constraint $\varepsilon$ for reliability. Due to the non-convexity of the objective function \eqref{prob_obj_func}, it is difficult to obtain the optimal solutions for the problem. The optimal solutions may be obtained by a two-dimensional (2D) exhaustive search but the computational complexity to search all possible transmission rates and minimum elevation angles is significantly high.
To reduce the computational complexity, we transform the problem with bounded feasible regions and use an iterative algorithm as described next.

\subsection{Problem Transformation and Iterative Algorithm}
We first simplify the problem \eqref{prob1} by transforming the constraints \eqref{prob_const_a} and \eqref{prob_const_b}.
From the definition of $\Pvis$ with $\Pinv$ in Lemma \ref{lem:Prob}, the constraint (\ref{prob_const_a}) can be rewritten as
\begin{align}\label{prob_const_a_2}
    \dmax\ge \sqrt{a^2 + 4\re (\re +a)(1- (1 - \eta)^{\frac{1}{S}})}.
\end{align}
Since the derivative of $\themin$ with respect to $\dmax$, given by (\ref{eq:1st_der}) shown at the top of this page, is always negative, $\themin$ is a monotonically decreasing function of $\dmax$.
\setcounter{equation}{56}
Therefore, by combining \eqref{prob_const_a_2} with \eqref{eq:law_of_cosines}, we have
\begin{align}\label{eq:mu}
    \themin \le \sin^{-1}\left(\frac{2(\re +a)(1 - \eta)^{\frac{1}{S}} - 2r_
    {\mathrm{e}} - a}{\sqrt{a^2 + 4\re (\re +a)(1- (1 - \eta)^{\frac{1}{S}})}}\right) \delequal \mu.
\end{align}
We now decompose the problem (\ref{prob1}) into two sub-problems.
For a given $\themin=\bar{\theta}_{\mathrm{min}}$, the maximum value of feasible $R$, denoted by $R_{\mathrm{max}} (\bar{\theta}_{\mathrm{min}})$, can be obtained at a point where the equality of (\ref{prob_const_b}) holds, because the outage probability increases with $R$.
With the bounded feasible region, $0\le R \le R_{\mathrm{max}}(\bar{\theta}_{\mathrm{min}})$, we can solve the following problem:
\begin{subequations}\label{prob_sub1}
\begin{align} 
{\mathop {{\mathop{\mathrm{maximize}}\nolimits} }\limits_{{R}} }&\,\,\,\,\,T(R,\bar{\theta}_{\mathrm{min}}) \label{prob_sub1_obj_func}\\
{{\mathrm{subject \,\,to}}}&\,\,\,\,\,{0 \le {R} \le R_{\mathrm{max}} (\bar{\theta}_{\mathrm{min}}).}\label{prob_sub1_const_a}
\end{align}
\end{subequations}
Let ${R}^*$ denote the optimal solution of the problem (\ref{prob_sub1}). 
Then, for a given $R={R}^*$, the outage constraint \eqref{prob_const_b} gives a lower bound of feasible $\themin$, denoted by $\theta_{0} ({R}^*)$, because the outage probability decreases with $\themin$. With the bounded feasible region, $\themin$ can be optimized as
\begin{subequations}\label{prob_sub2}
\begin{align} 
{\mathop {{\mathop{\mathrm{maximize}}\nolimits} }\limits_{{\themin}} }&\,\,\,\,\,T({R}^*,\themin) \label{prob_sub2_obj_func}\\
{{\mathrm{subject \,\,to}}}&\,\,\,\,\,{{\theta_{0} ({R}^*) \le \themin} \le \mu.} \label{prob_sub2_const_a}
\end{align}
\end{subequations}

We propose an algorithm iteratively solving the problems \eqref{prob_sub1} and \eqref{prob_sub2} as illustrated in Algorithm \ref{algorithm}.
We first calculate the upper bound of $\themin$, $\mu$, using \eqref{eq:mu} and then initialize $\themin=\mu$ and $T_{\mathrm{max}}=0$. 
For a given $\themin=\mu$, we calculate $R_{\mathrm{max}}$ from \eqref{prob_const_b} and find the optimal $R^*$ using a numerical search within $R\in[0,R_{\mathrm{max}}]$.
Then, for a fixed $R=R^*$, we calculate the lower bound of $\themin$, $\theta_0$, from \eqref{prob_const_b} and obtain the optimal $\themin^*$ using a numerical search within $\themin\in [\theta_0,\mu]$.
We compute the throughput with the updated $R^*$ and $\themin^*$ and compare it with $T_{\mathrm{max}}$. We update $T_{\mathrm{max}}=T(R^*, \themin^*)$ and repeat alternatively solving the two subproblems until $T(R^*, \themin^*) \le T_{\mathrm{max}}$.

\begin{algorithm}[t]
\begin{algorithmic}[1]
\caption{Iterative algorithm for throughput maximization} \label{algorithm}
    \REQUIRE $b$, $m$, $\Omega$, $f_{\mathrm{c}}$, $S$, $a$, $\re $, $\alpha$, $c$, $P$, $N_0$, $W$, $g$, $\Gtml$, $\Gtsl$, $\omega_{\mathrm{th}}$, $\Grvmax$, $\omega_{\mathrm{e}}$, $\eta$, and $\varepsilon$.
    \ENSURE $R^*$ and $\themin^{*}$.
        \STATE Calculate $\mu$ using (\ref{eq:mu}).
        \STATE Initialize $\themin^{(0)}=\mu$, which satisfies (\ref{prob_const_b}), $\hat{T}=0$, and $i=0$.
        \REPEAT
            \STATE $T_{\mathrm{max}}\leftarrow \hat{T}$
            \STATE Calculate $R_{\mathrm{max}}$ such that $\Pout(R_{\mathrm{max}}, {\theta}_{\mathrm{min}}^{(i)})=\varepsilon$.
            \STATE $R^{(i+1)} \leftarrow {\mathop{\mathrm{argmax}}\nolimits}_{0 \le {R} \le R_{\mathrm{max}}} T(R,\themin^{(i)})$
            \STATE Calculate $\theta_0$ such that $\Pout(R^{(i+1)},\theta_0)=\varepsilon$.
            \STATE $\themin^{(i+1)} \leftarrow {\mathop{\mathrm{argmax}}\nolimits}_{\theta_0 \le \themin \le \mu} T(R^{(i+1)},\themin)$
            \STATE Calculate $\hat{T}=T(R^{(i+1)},\themin^{(i+1)})$.        
            \STATE $i \leftarrow i + 1$ 
        \UNTIL{$\hat{T} > T_{\mathrm{max}}$}
        \STATE $R^* \leftarrow R^{(i-1)}$
        \STATE $\themin^* \leftarrow \themin^{(i-1)}$
\end{algorithmic}
\end{algorithm}

\subsection{Complexity Analyses}
The computational complexity to solve the problem \eqref{prob1} highly depends on the analytical expressions of the probabilities for the serving satellite's distributions and the outage probability.
The exact probabilities in Lemma \ref{lem:Prob} require the complexity of $\mathcal{O}(S)$, while the approximated ones in Lemma \ref{lem:prob_approx} require $\mathcal{O}(1)$ since they do not include terms with the $S$-th power but exponential functions instead.
In (\ref{eq:Poutml_fin}) and (\ref{eq:Poutsl_fin}), the combination $\binom{S-1}{k}$ and the terms such as $(a+2\re)^{2(S-1-k)}(-1)^{k}$ and $\dth^{2(k+1)}$ require $\mathcal{O}(S)$.
Then, the complexity of $\Pout$ is given by $\mathcal{O}(S^2N^2\tau)$ where $\tau$ is the complexity for the lower incomplete Gamma function.
Similarly, the complexity of $\Poutapp$ is obtained as $\mathcal{O}(N^2\hat{\tau})$ where $\hat{\tau}$ is the computational complexity for $\Cml[n]$ and $\Csl[n]$.
When $\alpha=2$, the complexity of $\Poutapp$ becomes $\mathcal{O}(N^2\tau)$.
Therefore, $\Poutapp$ is significantly less complex than $\Pout$ especially for large~$S$, since the complexity of $\Poutapp$ does not depend on $S$.

Let $\Delta_R$ and $\Delta_\theta$ denote the search steps for the optimization variables $R$ and $\themin$, respectively. 
To obtain the computational complexity of the 2D exhaustive search, we use sufficiently large $\hat{R}$ as an upper-bound of the search region since the feasible region of $R$ is not bounded in the problem \eqref{prob1}.
Then, the complexity of the 2D exhaustive search is given by $\mathcal{O}(c_{\mathrm{out}} \lfloor {90^\circ /{\Delta _\theta }} \rfloor \lfloor {\hat{R}/{\Delta _R}} \rfloor )$ where $c_{\mathrm{out}}$ is the complexity for calculating the outage probability.
The complexity of Algorithm \ref{algorithm} is given by $\mathcal{O}(c_{\mathrm{out}} L (\left\lfloor {R_{\mathrm{max}}(\mu )/{\Delta _R}} \right\rfloor  + \left\lfloor {\mu /{\Delta _\theta }} \right\rfloor ))$ where $L$ is the number of iterations of Algorithm \ref{algorithm}. 
This proves that the approximated expression and iterative algorithm make it much easier to obtain the solution of the throughput maximization problem.

\begin{table}
    \caption{Simulation Parameters}\label{Table:Sim_Param}
    \centering
    \begin{tabular}{|l|c|c|}
     \hline 
     Parameter & VSAT & Handheld \\
     \hline\hline
     Carrier frequency $f_{\mathrm{c}}$ [GHz] & 20 & 2  \\
     \hline
     Radius of the earth $r_{\mathrm{e}}$ [km] & \multicolumn{2}{c|}{6,378}\\
     \hline
     Path-loss exponent $\alpha$ & \multicolumn{2}{c|}{2}\\
     \hline
     Speed of light $c$ [m/s] & \multicolumn{2}{c|}{$3 \times 10^8$}\\
     \hline
     Minimum elevation angle $\theta_{\mathrm{min}}$ [deg] & \multicolumn{2}{c|}{$10$}\\
     \hline
     Noise spectral density $N_0$ [dBm/Hz] & \multicolumn{2}{c|}{$-174$}\\
     \hline
     Threshold angle between main/side lobes [deg] & \multicolumn{2}{c|}{$20$}\\
     \hline
     Transmit antenna gain for main lobes $\Gtml$ [dBi] & 38.5 & 30\\
     \hline
     Transmit antenna gain for side lobes $\Gtsl$ [dBi] & 28.5 & 20\\
     \hline
     EIRP density [dBW/MHz] & 4 & 34\\
     \hline
     Maximum receive antenna gain $\Grvmax$ or $\Grh$ [dBi] & 39.7 & 0 \\
     \hline
     Bandwidth $W$ [MHz] & 100 & 10 \\
     \hline
    \end{tabular}
\end{table}

\section{Numerical Results}\label{sec:sim}
In this section, we numerically verify the derived results based on the simulation parameters listed in Table \ref{Table:Sim_Param} unless otherwise stated.
The VSATs are considered for the Ka-band and the handheld terminals are targeted for the S-band as in the 3GPP standardization [\ref{Ref:3GPP_38.821}].
The effective isotropically radiated power (EIRP) density is calculated as
$P\Gtml/W$, from which the transmit power of the satellites can be obtained.
Three different shadowed-Rician fading models are taken into consideration: frequent heavy shadowing (FHS) $\{b=0.063,\ m=0.739,\ \Omega=8.97\times10^4\}$, average shadowing (AS) $\{b=0.126,\ m=10.1,\ \Omega=0.835\}$, and infrequent light shadowing (ILS) $\{b=0.158,\ m=19.4,\ \Omega=1.29\}$~[\ref{Ref:Abdi}].

Fig. \ref{Fig:Pcases_vs_S} shows the probabilities of three distribution cases for the serving satellite versus the number of satellites $S$.
As expected, for small $S$, the satellite-invisible probability is larger than the others, while as $S$ increases, the serving satellite is more probably located in $\Aml$. 
As $S$ increases, the probability that the serving satellite is in $\Asl$ first increases and then decreases.
This is because the surface area of $\Asl$, $\mathcal{S}({\Asl})=1.14 \times 10^7$ $\mathrm{km}^2$, is much larger than that of $\Aml$, $\mathcal{S}({\Aml})=1.82 \times 10^5$ $\mathrm{km}^2$.
From this fact, for small $S$, the serving satellite is more likely to be located in $\Asl$, while, for large $S$, $\Aml$ more probably includes the serving satellite.

Fig. \ref{Fig:Pvis_vs_theta_min} shows the satellite-visible probability $\Pvis$ versus the minimum elevation angle $\themin$ for various altitudes $a=\{300, 600, 1200\}$ km with $S=100$.
The exact expression of the satellite-visible probability perfectly matches the simulation results, and the approximated expression is very close to the exact one. 
As $\themin$ increases or $a$ decreases, the satellite-visible probability decreases because the surface area of $\Avis$ becomes smaller.
The required minimum elevation angle to achieve a certain level of satellite-visibility can be obtained from Fig. \ref{Fig:Pvis_vs_theta_min}.
For example, to achieve 90\% satellite-visible probability with $S=100$, $\themin$ should be less than or equal to $\{7.7, 20.7\}$ degrees for the altitude $\{600, 1200\}$ km, respectively, while this cannot be achieved with any $\themin$ for the altitude of $300$~km.

\begin{figure}
\begin{center}
\includegraphics[width=.885\columnwidth]{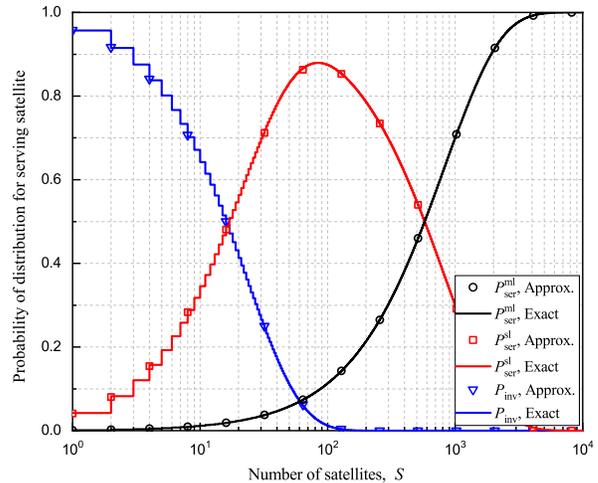}
\end{center}
\setlength\abovecaptionskip{.25ex plus .125ex minus .125ex}
\setlength\belowcaptionskip{.25ex plus .125ex minus .125ex}
\caption{Probabilities of the serving satellite's distributions versus the number of satellites $S$ with $a=1200$ km.}
\vspace{10pt}
\label{Fig:Pcases_vs_S}
\end{figure}

\begin{figure}
\begin{center}
\includegraphics[width=.885\columnwidth]{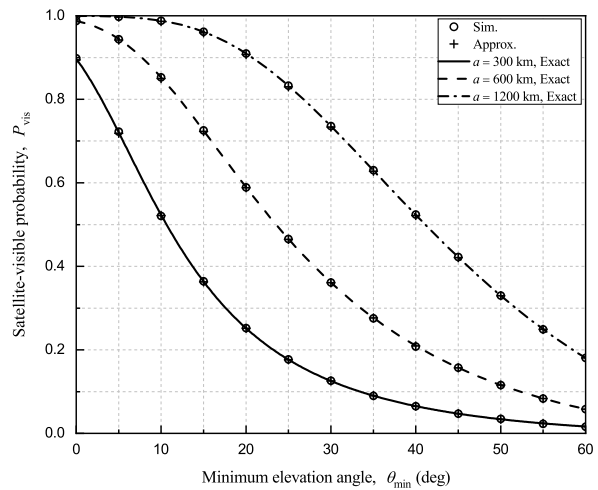}
\end{center}
\setlength\abovecaptionskip{.25ex plus .125ex minus .125ex}
\setlength\belowcaptionskip{.25ex plus .125ex minus .125ex}
\caption{Satellite-visible probability versus the minimum elevation angle $\themin$ for various altitudes $a=\{300, 600, 1200\}$ km with $S=100$.}
\vspace{10pt}
\label{Fig:Pvis_vs_theta_min}
\end{figure}

\begin{figure*}
\centering
\subfigure[VSAT]{
\includegraphics[width=.885\columnwidth]{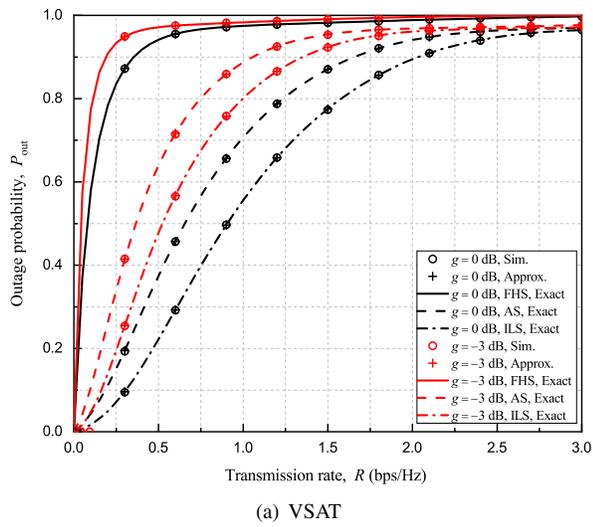}
\label{Fig:Pout_vs_R_VSAT}
}
\hspace{25pt}
\subfigure[Handheld terminal]{
\includegraphics[width=.885\columnwidth]{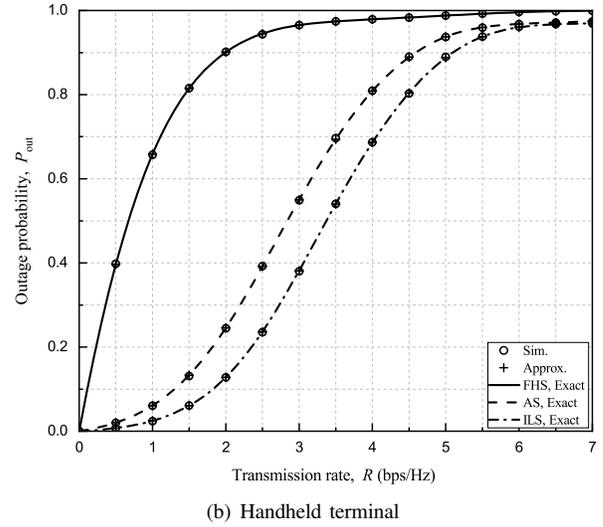}
\label{Fig:Pout_vs_R_HH}
}
\caption{Outage probability versus the transmission rate for various shadowed-Rician fading scenarios with $a=600$ km and $S=100$.}
\label{Fig:Pout_vs_R}
\end{figure*}

\begin{figure*}
\centering
\subfigure[VSAT]{
\includegraphics[width=.885\columnwidth]{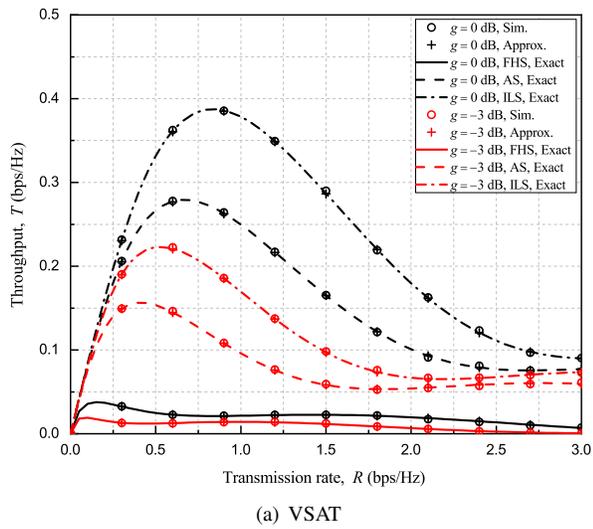}
\label{Fig:Thru_vs_R_VSAT}
}
\hspace{25pt}
\subfigure[Handheld terminal]{
\includegraphics[width=.885\columnwidth]{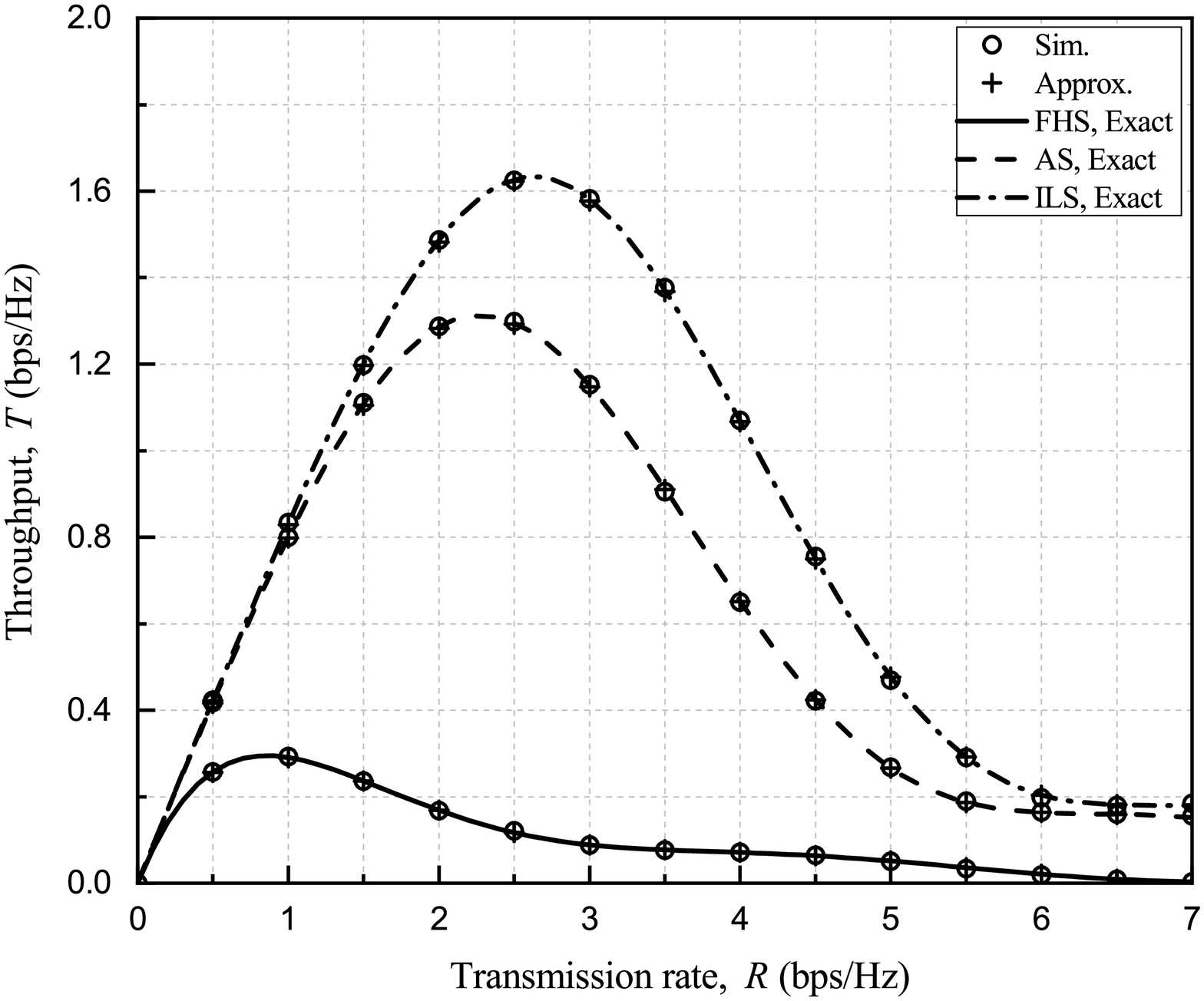}
\label{Fig:Thru_vs_R_HH}
}
\caption{System throughput versus the transmission rate for various shadowed-Rician fading scenarios with $a=600$ km and $S=100$.}
\label{Fig:Thru_vs_R}
\end{figure*}

\begin{figure}
\begin{center}
\includegraphics[width=.885\columnwidth]{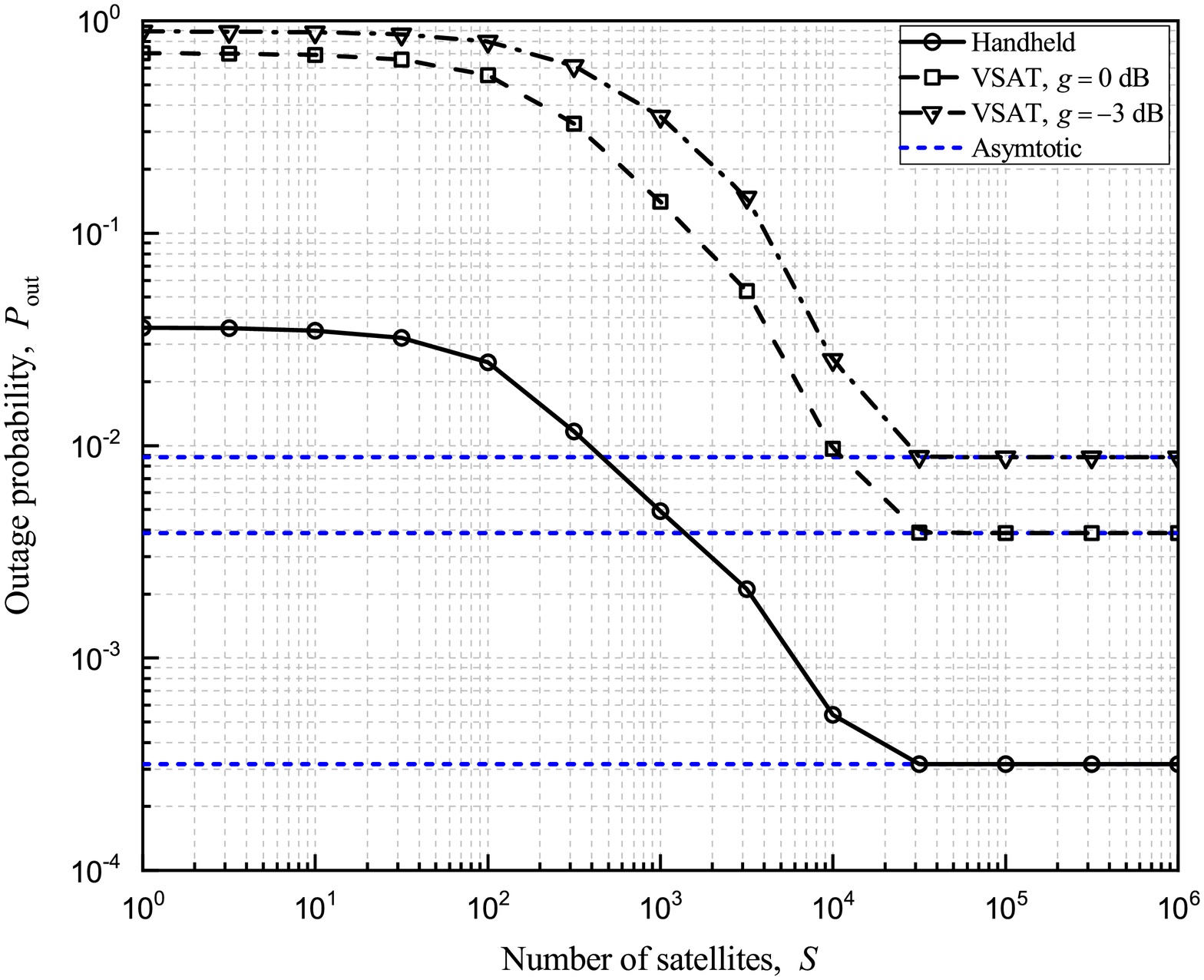}
\end{center}
\setlength\abovecaptionskip{.25ex plus .125ex minus .125ex}
\setlength\belowcaptionskip{.25ex plus .125ex minus .125ex}
\caption{Outage probability versus the number of satellites $S$ for the ILS with $a=600$ km and $R=1$ bps/Hz.}
\vspace{10pt}
\label{Fig:Pout_vs_S}
\end{figure}

\begin{figure}
\begin{center}
\includegraphics[width=.885\columnwidth]{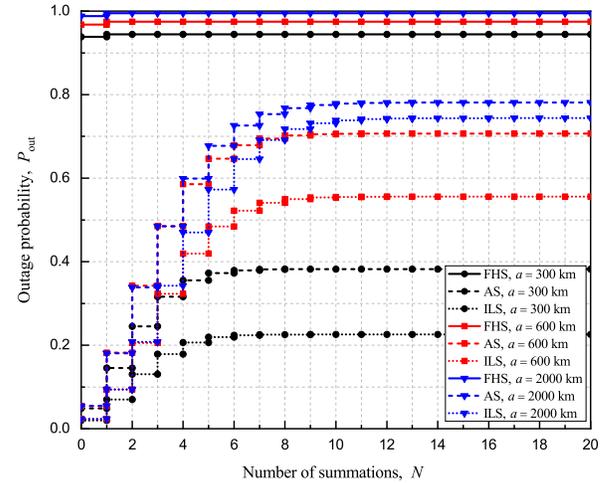}
\end{center}
\setlength\abovecaptionskip{.25ex plus .125ex minus .125ex}
\setlength\belowcaptionskip{.25ex plus .125ex minus .125ex}
\caption{Outage probability versus the limited number of the summations in \eqref{eq:CDF_ch_gain}, $N$, for the VSAT with $g=0$ dB and $S=100$.}
\vspace{10pt}
\label{Fig:Pout_vs_N}
\end{figure}

\begin{figure}
\begin{center}
\includegraphics[width=.885\columnwidth]{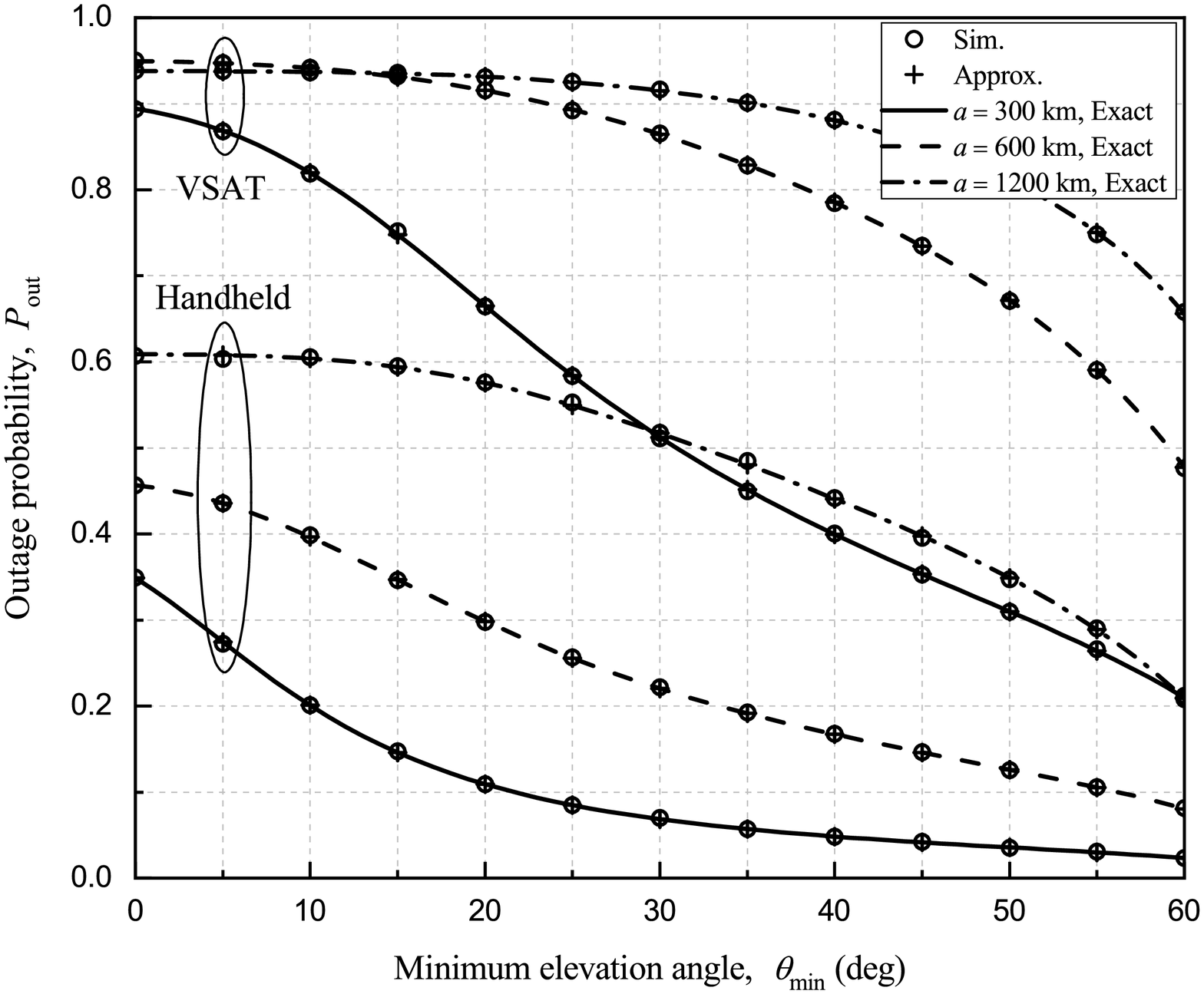}
\end{center}
\setlength\abovecaptionskip{.25ex plus .125ex minus .125ex}
\setlength\belowcaptionskip{.25ex plus .125ex minus .125ex}
\caption{Outage probability versus the minimum elevation angle $\themin$ for various altitudes $a=\{300, 600, 1200\}$ km and the FHS with $S=100$, $R=0.5$ bps/Hz, and $g=-3$ dB.}
\vspace{10pt}
\label{Fig:Pout_vs_theta_min}
\end{figure}

Fig. \ref{Fig:Pout_vs_R} shows the outage probability of the system versus the transmission rate $R$ for the VSAT and handheld terminal under various fading scenarios.
The analytical results of the exact and approximated outage probabilities are obtained from Theorems \ref{thm:Pout} and \ref{thm:Poutapprox}, respectively.
The exact outage probability well matches the simulation results, and the approximated outage probability is also very close.
As $R$ increases, the outage probability increases, which is expected from the definition of the outage probability. 
As the channel experiences more severe shadowing, the outage probability becomes larger due to the lower received SNR.
For the VSAT, as rain attenuation becomes severe, the outage probability increases, which is a major drawback of the Ka band instead of using wide bandwidth.

Fig. \ref{Fig:Thru_vs_R} shows the system throughput versus the transmission rate $R$ for the VSAT and handheld terminal under various fading scenarios.
The throughput curves obtained from the exact analysis match the simulation results, and the approximated results are almost the same as the exact ones. 
The handheld terminal has the higher system throughput than the VSAT due to the better characteristics of the S-band.
However, since we can usually use wider bandwidth in the Ka-band compared to the S-band, the throughput in bps of the VSAT can be much higher than that of the handheld terminal.

Fig. \ref{Fig:Pout_vs_S} shows the outage probability versus the number of satellites $S$ for the ILS scenario with $a=600$ km and $R=1$ bps/Hz where the asymptotic bounds are from Corollary 2.
As $S$ increases, the outage probability first decreases and then stays constant.
This is because as $S$ increases, the distance to the serving satellite decreases and then becomes close to the minimum distance $a$, as proved in Corollary \ref{cor:PoutapproxAsym}.

Fig. \ref{Fig:Pout_vs_N} shows the outage probability when the number of the summations in \eqref{eq:CDF_ch_gain} is limited to $N$. 
The outage probability converges as $N$ increases, as proven in Section~\ref{sec:app}.
Since a few tens of summations seem to be enough to converge, the computational complexity for the exact outage probability, i.e., $\mathcal{O}(S^2N^2\tau)$, is mainly affected by the number of satellites $S$.
In addition, the complexity of the approximated outage probability, i.e., $\mathcal{O}(N^2\tau)$, can be very small.

Fig. \ref{Fig:Pout_vs_theta_min} shows the outage probability $\Pout$ versus the minimum elevation angle $\themin$ for various altitudes $a=\{300, 600, 1200\}$ km with $S=100$, $R=0.5$ bps/Hz, and $g=-3$ dB.
As $\themin$ increases, the outage probability decreases because the channel quality between the terminal and the serving satellite becomes better.
It is shown that the handheld terminal has better outage performance than the VSAT because the S-band has less limitation on the EIRP density and experiences less path-loss compared to the Ka-band.

Fig. \ref{Fig:Thrumax_vs_S} shows the maximum system throughput for the VSATs with beam-pointing errors $\omega_\mathrm{e}=\{0,1\}^\circ$, $g=\{-3,0\}$ dB, $a=600$ km, $\eta=0.9$, and $\varepsilon=0.1$.
The optimal system throughput $T^*$ is obtained using 2D exhaustive search, while the sub-optimal solutions are given by Algorithm \ref{algorithm}.
The proposed algorithm has very close performance to the 2D exhaustive search.
The maximum system throughput also increases with $S$ because the satellite-visible probability and the outage probability are increasing and decreasing functions of $S$, respectively.
The high beam-pointing error, e.g., $1^\circ$, and sever rain attenuation, e.g., $g=-3$ dB, degrade the quality of received signals so that the system throughput is reduced.

\begin{figure}
\begin{center}
\includegraphics[width=.885\columnwidth]{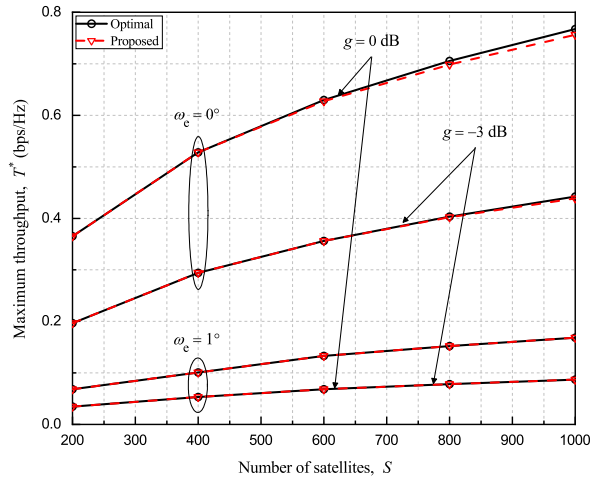}
\end{center}
\setlength\abovecaptionskip{.25ex plus .125ex minus .125ex}
\setlength\belowcaptionskip{.25ex plus .125ex minus .125ex}
\caption{Maximum system throughput versus the number of satellites $S$ for VSAT with $g=\{-3, 0\}$ dB, $\omega_{\mathrm{e}}=\{0,1\}^\circ$, $a=600$ km, $\eta=0.9$, and $\varepsilon=0.1$.}
\label{Fig:Thrumax_vs_S}
\end{figure}

\section{Conclusions}\label{sec:con}
In this paper, we considered downlink LEO satellite communication systems where multiple LEO satellites are uniformly distributed at a certain altitude. 
We analyzed the distance distributions and the probabilities of distribution cases for the serving satellite.
We derived the exact outage probability, and the approximated expression was obtained by using the Poisson limit theorem. 
With the derived expressions, we optimized the transmission rate and the minimum elevation angle to maximize the system throughput by the proposed iterative algorithm.
The complexity of the proposed algorithm and exhaustive search was compared.
Simulation results verified the exact and approximated analyses and showed that the proposed algorithm has close performance to the optimum.
The approximated expressions are expected to be used with very high accuracy but low complexity to analyze satellite communication systems having more than thousands of satellites, which will be realized in the near future.

\appendices
\section{Proof of Lemma \ref{lem:Prob}}\label{app:lem_prob}
Using the success probability, the void probability for a region $\mathcal{X}$, i.e., the probability that there is no satellite in $\mathcal{X}$, is given by [\ref{Ref:Book:Chiu}]
\begin{align}\label{eq:void_prob}
\PvoidBPP(\mathcal{X})
    = \prod_{s\in\BPP}\left(1-\frac{\mathcal{S}(\mathcal{X})}{\SBPP}\right)
    = \left(1-\frac{\mathcal{S}(\mathcal{X})}{\SBPP}\right)^S.
\end{align}
The probability of Case 1 is the probability that there is at least one satellite in $\Aml$, which is given by 
\begin{align}
\Pserml
    = 1 - \P[\BPPml = \emptyset]
    = 1 - \PvoidBPP(\Aml).
\end{align}
The probability of Case 2 is the probability that there is no satellite in $\Aml$ but at least one satellite in $\Asl$, which is given by
\begin{align}
\Psersl
    &=\P[\BPPml = \emptyset, \BPPsl \neq \emptyset]\nonumber\\
    &= \P[\BPPml = \emptyset] - \P[\BPPml = \emptyset, \BPPsl = \emptyset]\nonumber\\
    &=\PvoidBPP(\Aml) - \PvoidBPP(\Avis).
\end{align}
The probability of Case 3 is called the satellite-invisible probability, i.e., the probability that all satellites are in $\Avis^\mathrm{c}$, which is given by 
\begin{align}
\Pinv
    =\P[\BPPvis = \emptyset]
    = \PvoidBPP(\Avis).
\end{align}
Using \eqref{eq:void_prob} with the surface areas obtained in Section \ref{sec:distDist:surf}, the probabilities of the three cases can be obtained.

\section{Proof of Theorem \ref{thm:Poutapprox}}\label{app:thm_pout_approx}
Similar to the derivation of the exact outage probability, the approximated outage probability of the system is given in \eqref{eq:Poutapp_fin}.
Using $\bar{f}_Y(x)$, $\Poutmlapp$ in \eqref{eq:Poutapp_fin} is given by
\begin{align}\label{eq:Poutmlapp1}
\Poutmlapp
    & = \int_{a}^{\dth} F_{h_{\s0}}(w_1 x^{\alpha}) \bar{f}_Y(x)dx \nonumber\\
    & = \frac{K S}{2\re(\re+a)\bar{F}_D(\dth)}\sum\limits_{n = 0}^\infty  {\frac{{{{(m)}_n}{\delta ^n}{{(2b)}^{1 + n}}}}{{{{(n!)}^2}}}} \nonumber\\
    &\qt  \int_{a}^{\dth} \gamma\left(1+n,\frac{w_1 x^{\alpha}}{2b}\right)  e^{-\kappa(x)} x dx.
\end{align}
By letting $\kappa(x)=\frac{S(x^2-a^2)}{4\re(\re+a)} \delequal y$, the integral in \eqref{eq:Poutmlapp1} can be expressed as 
\begin{align}\label{eq:integral2}
&\int_{a}^{\dth} \int_0^{\frac{w_1 x^{\alpha}}{2b}} t^n e^{-t-\kappa(x)} x dt dx \nonumber\\
    &= \frac{2\re(\re+a)}{S} \!\!\int_{0}^{\kappa(\dth)} \int_0^{\frac{w_1}{2b}(\kappa^{-1}(y))^\alpha} t^n e^{-t-y} dt dy
\end{align}
where $\kappa^{-1}(y)=\sqrt{\frac{4\re(\re+a)y}{S}+a^2}$.
The domain of the integration in (\ref{eq:integral2}) is shown as the shaded areas in Fig. \ref{Fig:Region2} and can be divided into two domains $\mathcal{D}_3$ and $\mathcal{D}_4$, which are given by 
$
\mathcal{D}_3=  \left\{0 \le t \le \frac{w_1 a^{\alpha}}{2b}, 0 \le y \le \kappa(\dth)\right\}    
$
and 
$
\mathcal{D}_4= \left\{\frac{w_1 a^{\alpha}}{2b} \!\le t \le \!\frac{w_1 \dth^{\alpha}}{2b}, \kappa\left(\left(\frac{2 b t}{w_1}\right)^{\frac{1}{\alpha}}\right) \le \!y\! \le \kappa(\dth)\right\}
$,
respectively.
The integral over the domain $\mathcal{D}_3$ is given by 
\begin{align}\label{eq:integral_s3}
\mathcal{I}_{\mathcal{D}_3}
    &= \int_{0}^{\kappa(\dth)} e^{-y} dy \times \int_0^{\frac{w_1 a^{\alpha}}{2b}} t^n e^{-t} dt\nonumber\\
    &= (1-e^{-\kappa(\dth)}) \gamma\left(1+n,\frac{w_1 a^{\alpha}}{2b}\right),
\end{align}
and the integral over the domain $\mathcal{D}_4$ is given by
\begin{align}\label{eq:integral_s4}
\mathcal{I}_{\mathcal{D}_4}
    &=\int_{\frac{w_1 a^{\alpha}}{2b}}^{\frac{w_1 \dth^{\alpha}}{2b}} \int_{\kappa\left(\left(\frac{2 b t}{w_1}\right)^{1/\alpha}\right)}^{\kappa(\dth)} t^n e^{-t-y} dy dt \nonumber\\
    &=\int_{\frac{w_1 a^{\alpha}}{2b}}^{\frac{w_1 \dth^{\alpha}}{2b}} t^n e^{-t} \left(e^{-\kappa\left(\left(\frac{2 b t}{w_1}\right)^{1/\alpha}\right)} - e^{-\kappa(\dth)}\right) dt \nonumber\\
    &=\int_{\frac{w_1 a^{\alpha}}{2b}}^{\frac{w_1 \dth^{\alpha}}{2b}} t^n e^{-t-\kappa\left(\left(\frac{2 b t}{w_1}\right)^{1/\alpha}\right)}dt - e^{-\kappa(\dth)}\int_{\frac{w_1 a^{\alpha}}{2b}}^{\frac{w_1 \dth^{\alpha}}{2b}} t^n e^{-t} dt \nonumber\\
    &=e^{\frac{S a^2}{4\re(\re+a)}} \mathcal{C}_\mathrm{ml}[n] - e^{-\kappa(\dth)} \nonumber\\
    &\qt \left(\gamma\left(1+n, \frac{w_1 \dth^{\alpha}}{2b}\right) - \gamma\left(1+n, \frac{w_1 a^{\alpha}}{2b}\right)\right)
\end{align}
where $\mathcal{C}_\mathrm{ml}[n]$ is given in \eqref{eq:Cml}. From \eqref{eq:Poutmlapp1}, \eqref{eq:integral2}, \eqref{eq:integral_s3}, and \eqref{eq:integral_s4}, the final expression of $\Poutmlapp$ can be obtained. 
The derivation of $\Poutslapp$ can be done by the similar steps as those of $\Poutmlapp$, so omitted due to space limitation. Using $\Poutmlapp$ and $\Poutslapp$ with the results in Lemma \ref{lem:prob_approx}, the final expression in Theorem \ref{thm:Poutapprox} can be obtained.

\begin{figure}
\centering
\includegraphics[width=.8\columnwidth]{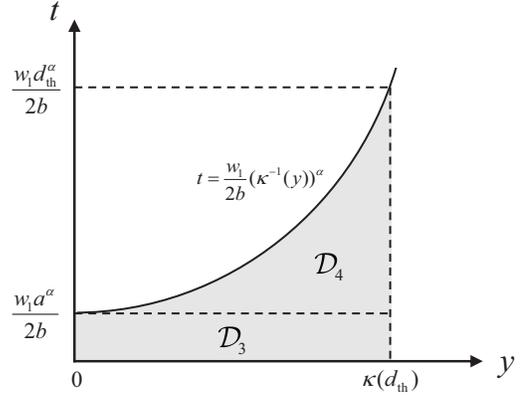}
\caption{Domain of the integration in \eqref{eq:integral2}.}
\label{Fig:Region2}
\end{figure}

\ifCLASSOPTIONcaptionsoff
  \newpage
\fi


\begin{thebibliography}{1}


\bibitem{bib:3GPP_38.811}\label{Ref:3GPP_38.811}
3GPP TR 38.811 v15.4.0, ``Study on NR to support non-terrestrial networks," Sep. 2020.

\bibitem{bib:3GPP_38.821}\label{Ref:3GPP_38.821}
3GPP TR 38.821 v16.0.0, ``Solutions for NR to support non-terrestrial networks (NTN)," Dec. 2019.

\bibitem{bib:Lin1}\label{Ref:Lin1}
Z. Lin, M. Lin, B. Champagne, W.-P. Zhu, and N. Al-Dhahir, ``Secure beamforming for cognitive satellite terrestrial networks with unknown eavesdroppers," \emph{IEEE Syst. J.}, vol. 15, no. 2, pp. 2186-2189, June 2021.

\bibitem{bib:Lin3}\label{Ref:Lin2}
Z. Lin, M. Lin, T. de Cola, J.-B. Wang, W.-P. Zhu, and J. Cheng, ``Supporting IoT with rate-splitting multiple access
in satellite and aerial-integrated networks," \emph{IEEE Internet Things J.}, vol. 8, no. 14, pp. 11123-11134, July 2021.

\bibitem{bib:Lin4}\label{Ref:Lin3}
Z. Lin, M. Lin, B. Champagne, W.-P. Zhu, and N. Al-Dhahir, ``Secure and energy efficient transmission for RSMA-based cognitive satellite-terrestrial networks," \emph{IEEE Wireless Commun. Lett.}, vol. 10, no. 2, pp. 251-255, Feb. 2021.

\bibitem{bib:Guidotti}\label{Ref:Guidotti}
A. Guidotti, A. Vanelli-Coralli, M. Conti, S. Andrenacci, S. Chatzinotas, N. Maturo, B. Evans, A. Awoseyila, A. Ugolini, T. Foggi, L. Gaudio, N. Alagha, and S. Cioni, ``Architectures and key technical challenges for 5G systems incorporating satellites," \emph{IEEE Trans. Veh. Technol.}, vol. 68, no. 3, pp. 2624-2639, Mar. 2019.

\bibitem{bib:Shahid}\label{Ref:Shahid}
S. M. Shahid, Y. T. Seyoum, S. H. Won, and S. Kwon, ``Load balancing for 5G integrated satellite-terrestrial networks," \emph{IEEE Access}, vol. 8, pp. 132144-132156, July 2020.

\bibitem{bib:Zhen}\label{Ref:Zhen}
L. Zhen, T. Sun, G. Lu, K. Yu, and R. Ding, ``Preamble design and detection for 5G enabled satellite random access," \emph{IEEE Access}, vol. 8, pp. 49873-49884, Mar. 2020.

\bibitem{bib:Loo}\label{Ref:Loo}
C. Loo, ``A statistical model for a land mobile satellite link," \emph{IEEE Trans. Veh. Technol.}, vol. 34, no. 3, pp. 122-127, Dec. 1985.

\bibitem{bib:Abdi}\label{Ref:Abdi}
A. Abdi, W. C. Lau, M.-S. Alouini, and M. Kaveh, ``A new simple model for land mobile satellite channels: First- and second-order statistics," \emph{IEEE Trans. Wireless Commun.}, vol. 2, no. 3, pp. 519-528, May 2003.

\bibitem{bib:Jung}\label{Ref:Jung}
D.-H. Jung and D.-G. Oh, ``Outage performance of shared-band on-board processing satellite communication system," \emph{Proc. IEEE VTC 2018-Fall}, Chicago, IL, Aug. 2018.

\bibitem{bib:Bhatnagar}\label{Ref:Bhatnagar}
M. R. Bhatnagar and M. K. Arti, ``On the closed-form performance analysis of maximal ratio combining in shadowed-Rician fading LMS channels," \emph{IEEE Commun. Lett.}, vol. 18, no. 1, pp. 54-57, Jan. 2014.

\bibitem{bib:Bhatnagar2}\label{Ref:Bhatnagar2}
M. R. Bhatnagar, ``Performance evaluation of decode-and-forward satellite relaying," \emph{IEEE Trans. Veh. Technol.}, vol. 64, no. 10, pp. 4827-4833, Oct. 2015.

\bibitem{bib:Bankey}\label{Ref:Bankey}
V. Bankey, P. K. Upadhyay, D. B. Da Costa, P. S. Bithas, A. G. Kanatas, and U. S. Dias, "Performance analysis of multi-antenna multiuser hybrid satellite-terrestrial relay systems for mobile services delivery," \emph{IEEE Access}, vol. 6, pp. 24729-24745, Apr. 2018.

\bibitem{bib:An}\label{Ref:An}
K. An and T. Liang, "Hybrid satellite-terrestrial relay networks with adaptive transmission," \emph{IEEE Trans. Veh. Technol.}, vol. 68, no. 12, pp. 12448-12452, Dec. 2019.

\bibitem{bib:Andrews}\label{Ref:Andrews}
J. G. Andrews, F. Baccelli, and R. K. Ganti, ``A tractable approach to coverage and rate in cellular networks," \emph{IEEE Trans. Commun.}, vol. 59, no. 11, pp. 3122-3134, Nov. 2011.

\bibitem{bib:Dhillon}\label{Ref:Dhillon}
H. S. Dhillon, R. K. Ganti, F. Baccelli, and J. G. Andrews, ``Modeling and analysis of $K$-tier downlink heterogeneous cellular networks," \emph{IEEE J. Sel. Areas Commun.}, vol. 30, no. 3, pp. 550-560, Apr. 2012.

\bibitem{bib:Jo}\label{Ref:Jo}
H.-S. Jo, Y. J. Sang, P. Xia, and J. G. Andrews, ``Heterogeneous cellular networks with flexible cell association: A comprehensive downlink SINR analysis," \emph{IEEE Trans. Wireless Commun.}, vol. 11, no. 10, pp. 3484-3495, Oct. 2012.

\bibitem{bib:Singh}\label{Ref:Singh}
S. Singh, H. S. Dhillon, and J. G. Andrews, ``Offloading in heterogeneous networks: Modeling, analysis, and design insights," \emph{IEEE Trans. Wireless Commun.}, vol. 12, no. 5, pp. 2484-2497, May 2013.

\bibitem{bib:Kolawole}\label{Ref:Kolawole}
O. Y. Kolawole, S. Vuppala, M. Sellathurai, and T. Ratnarajah, ``On the performance of cognitive satellite-terrestrial networks," \emph{IEEE Trans. Cogn. Commun. Netw.}, vol. 3, no. 4, pp. 668-683, Dec. 2017.

\bibitem{bib:Guo}\label{Ref:Guo}
J. Guo, S. Durrani, and X. Zhou, ``Outage probability in arbitrarily-shaped finite wireless networks," \emph{IEEE Trans. Commun.}, vol. 62, no. 2, pp. 699-712, Feb. 2014.

\bibitem{bib:Chiu}\label{Ref:Book:Chiu}
S. N. Chiu, D. Stoyan, W. S. Kendall, and J. Mecke, \emph{Stochastic Geometry and Its Applications}, 3nd ed. New York, NY: Wiley, 2013.

\bibitem{bib:Okati}\label{Ref:Okati}
N. Okati, T. Riihonen, D. Korpi, I. Angervuori, and R. Wichman, ``Downlink coverage and rate analysis of low Earth orbit satellite constellations using stochastic geometry," \emph{IEEE Trans. Commun.}, vol. 68, no. 8, pp. 5120-5134, Aug. 2020.

\bibitem{bib:Talgat}\label{Ref:Talgat}
A. Talgat, M. A. Kishk, and M.-S. Alouini, ``Stochastic geometry-based analysis of LEO satellite communication systems," \emph{IEEE Commun. Lett.}, vol. 25, no. 8, pp. 2458-2462, Aug. 2021.

\bibitem{bib:Talgat2}\label{Ref:Talgat2}
A. Talgat, M. A. Kishk, and M.-S. Alouini, ``Nearest neighbor and contact distance distribution for binomial point process on spherical surfaces," \emph{IEEE Commun. Lett.}, vol. 24, no. 12, pp. 2659-2663, Dec. 2020.

\bibitem{bib:Huang1}\label{Ref:Huang1}
Q. Huang, M. Lin, W.-P. Zhu, S. Chatzinotas, and M.-S. Alouini, ``Performance analysis of integrated satellite-terrestrial multiantenna relay networks with multiuser scheduling," \emph{IEEE Trans. Aerosp. Electron. Syst.}, vol. 56, no. 4, pp. 2718-2731, Aug. 2020.

\bibitem{bib:Huang2}\label{Ref:Huang2}
Q. Huang, M. Lin, J.-B. Wang, T. A. Tsiftsis, and J. Wang, ``Energy efficient beamforming schemes for satellite-aerial-terrestrial networks
," \emph{IEEE Trans. Commun.}, vol. 68, no. 6, pp. 3863-3875, June 2020.

\bibitem{bib:Zhang}\label{Ref:Zhang}
J. Zhang, M. Lin, J. Ouyang, W.-P. Zhu, and T. de Cola, ``Robust beamforming for enhancing security in multibeam satellite systems," \emph{IEEE Commun. Lett.}, vol. 25, no. 7, pp. 2161-2165, July 2021.

\bibitem{bib:Alkhateeb}\label{Ref:Alkhateeb}
A. Alkhateeb, Y.-H. Nam, M. S. Rahman, J. Zhang, and R. W. Heath, Jr., ``Initial beam association in millimeter wave cellular systems: Analysis and design insights," \emph{IEEE Trans. Wireless Commun.}, vol. 16, no. 5, pp. 2807-2821, May 2017.

\bibitem{bib:Zhu}\label{Ref:Zhu}
Y. Zhu, G. Zheng, and M. Fitch, ``Secrecy rate analysis of UAV-enabled mmWave networks using Mat\'{e}rn hardcore point processes," \emph{IEEE J. Sel. Areas Commun.}, vol. 36, no. 7, pp. 1397-1409, July 2018.

\bibitem{bib:Dabiri}\label{Ref:Dabiri}
M. T. Dabiri, M. Rezaee, V. Yazdanian, B. Maham, W. Saad, and C. S. Hong, ``3D channel characterization and performance analysis of UAV-assisted millimeter wave links," \emph{IEEE Trans. Wireless Commun.}, vol. 20, no. 1, pp. 110-125, Jan. 2021.

\bibitem{bib:Na}\label{Ref:Na}
D.-H. Na, K.-H. Park, Y.-C. Ko, and M.-S. Alouini, ``Performance analysis of satellite communication systems with randomly located ground users," \emph{IEEE Trans. Wireless Commun.}, To be appeared.

\bibitem{bib:Zheng}\label{Ref:Zheng}
G. Zheng, S. Chatzinotas, and B. Ottersten, ``Generic optimization of linear precoding in multibeam satellite systems," \emph{IEEE Trans. Wireless Commun.}, vol. 11, no. 6, pp. 2308-2320, June 2012.

\bibitem{bib:Arti}\label{Ref:Arti}
M. K. Arti, ``Two-way satellite relaying with estimated channel gains," \emph{IEEE Trans. Commun.}, vol. 64, no. 7, pp. 2808-2820, July 2016.
    
\bibitem{bib:Guo2}\label{Ref:Guo2}
K. Guo, M. Lin, B. Zhang, J.-B. Wang, Y. Wu, W.-P. Zhu, and J. Cheng, ``Performance analysis of hybrid satellite-terrestrial cooperative networks with relay selection," \emph{IEEE Trans. Veh. Technol.}, vol. 69, no. 8, pp. 9053-9067, Aug. 2020.

\bibitem{bib:Polyanin}\label{Ref:Book:Polyanin}
A. D. Polyanin and A. V. Manzhirov, \emph{Handbook of Mathematics for Engineers and Scientists}, 1st ed. CRC Press, 2006. 

\bibitem{bib:Choudary}\label{Ref:Book:Choudary}
A. D. R. Choudary and C. P. Niculescu, \emph{Real Analysis on Intervals}, Springer, 2014.

\bibitem{bib:Nasir}\label{Ref:Nasir}
A. A. Nasir, X. Zhou, S. Durrani, and R. A. Kennedy, ``Relaying protocols for wireless energy harvesting and information processing," \emph{IEEE Trans. Wireless Commun.}, vol. 12, no. 7, pp. 3622-3636, July 2013.

\end{thebibliography}
\end{document}